\documentclass[cmp,numbook,envcountsame]{svjour}

\usepackage{amsmath} 
\usepackage{amssymb}
\usepackage{amsfonts}
\usepackage{graphicx}
\newcommand{\CC}{{\mathbb C}}
\newcommand{\RR}{{\mathbb R}}

\newcommand{\NN}{{\mathbb N}}
\newcommand{\TT}{{\mathbb T}}
\newcommand{\ZZ}{{\mathbb Z}}

\newcommand{\cB}{{\mathcal{B}}}

\newcommand{\cD}{{\mathcal{D}}}
\newcommand{\cF}{{\mathcal{F}}}
\newcommand{\cH}{{\mathcal{H}}}
\newcommand{\cK}{{\mathcal{K}}}

\newcommand{\bPhi}{{\mbox{\boldmath $\Phi$}}}
\newcommand{\bPsi}{{\mbox{\boldmath $\Psi$}}}
\newcommand{\bOmega}{{\mbox{\boldmath $\Omega$}}}

\newcommand{\bx}{{\mbox{\boldmath $x$}}}

\newcommand{\by}{{\mbox{\boldmath $y$}}}

\newcommand{\bK}{{\mbox{\boldmath $K$}}}

\newcommand{\bO}{{\mbox{\boldmath $O$}}}
\newcommand{\sbO}{{\mbox{\footnotesize \boldmath $O$}}}

\newcommand{\bP}{{\mbox{\boldmath $P$}}}
\newcommand{\sbP}{{\mbox{\footnotesize \boldmath $P$}}}
\newcommand{\bQ}{{\mbox{\boldmath $Q$}}}

\newcommand{\bfA}{{\mbox{\boldmath $\mathfrak A$}}}

\newcommand{\obfA}{\overline{\mbox{\boldmath $\mathfrak A$}}}

\newcommand{\fB}{{\mathfrak B}}
\newcommand{\fC}{{\mathfrak C}}

\newcommand{\bfF}{{\mbox{\boldmath $\mathfrak F$}}}
\newcommand{\obfF}{\overline{\mbox{\boldmath $\mathfrak F$}}}

\newcommand{\bfK}{{\mbox{\boldmath $\mathfrak K$}}}
\newcommand{\fK}{{\mathfrak K}}

\newcommand{\bfN}{{\mbox{\boldmath $\mathfrak N$}}}
\newcommand{\fR}{{\mathfrak R}}

\newcommand{\bfR}{{\mbox{\boldmath $\mathfrak R$}}}

\newcommand{\bpartial}{{\mbox{\boldmath $\partial$}}}

\newcommand{\eg}{{\it e.g.}}

\newcommand{\supp}{\text{supp}}
\newcommand{\scirc}{\mbox{\footnotesize $\circ$}}

\newcommand{\one}{{\mbox{\boldmath $1$}}}
\newcommand{\ad}{\text{Ad} \, }
\newcommand{\mybinom}[3][0.8]{\scalebox{#1}{$\dbinom{#2}{#3}$}}

\def\eg{{\it e.g.\ }}
\def\ie{{\it i.e.\ }}

\title{The resolvent algebra of non-relativistic Bose fields:
  sectors, morphisms, fields and dynamics}

\author{Detlev Buchholz}
\institute{Mathematisches Institut, 
Universit\"at G\"ottingen, \ 37073 G\"ottingen, Germany \\ 
\email detlev.buchholz@mathematik.uni-goettingen.de}

\authorrunning{Detlev Buchholz} 
\titlerunning{The resolvent algebra: sectors, morphisms, fields and dynamics}

\date{}

\begin{document}

\maketitle

\begin{abstract} 
\noindent
It was recently shown~\cite{Bu1} that the 
resolvent algebra of a non-relati\-vistic Bose field 
determines a gauge invariant (particle number preserving) 
kinematical algebra of observables 
which is stable under the automorphic action of a large family of 
interacting dynamics involving pair potentials. In the present article,   
this observable algebra is extended to a field algebra 
by adding to it isometries, 
which transform as tensors under gauge transformations and 
induce particle number changing morphisms of the observables.
Different morphisms are linked by intertwiners in the observable 
algebra. It is shown that such intertwiners also induce  
time translations of the morphisms. 
As a consequence, the field algebra is stable under the 
automorphic action of the interacting dynamics as well. These 
results establish a concrete C*-algebraic framework 
for interacting non-relativistic Bose systems in infinite 
space. It provides an adequate basis for studies of 
long range phenomena, such as phase transitions, 
stability properties of equilibrium states, condensates, 
and the breakdown of symmetries.
\end{abstract}
\section{Introduction}
\setcounter{equation}{0}

\noindent We continue here our study of the
stability properties of the resolvent algebra of a non-relativistic
Bose field under the action of interacting dynamics, involving a large 
family of pair potentials. It is our goal to establish an  
algebraic framework which allows one to treat interacting bosonic systems
in infinite volume without having to rely on finite
volume approximations. This seems desirable since it
puts one into the position to apply methods of the theory of operator 
algebras to specific problems in many body theory.
Conceivable applications are the
stability of equilibrium states,
phase transitions, moving systems such as flows, condensates
and the spontaneous breakdown of symmetries,
cf.~\cite{BrRo,CoDeZi,LiSeSoYn,PiSt,Ve}. 
Once the algebra, including the dynamics, has been constructed,
all states of interest appear as elements of its dual
space and are thus accessible to further study. 

\medskip 
A first step in this program was recently accomplished
in \cite{Bu1}. There it was shown that a slight extension of
the subalgebra of the resolvent algebra,
consisting of gauge invariant (particle number preserving) 
observables, is stable under the action of dynamics
involving pair potentials. It is the aim of the 
present article to extend this observable algebra to a larger field
algebra of operators, which change the particle numbers. 
Our approach is based on ideas
developed by Doplicher, Haag and Roberts in a 
general analysis of superselection sectors in relativistic quantum
field theory \cite{DoHaRo}. These ideas can be carried 
over to the non-relativistic setting with appropriate
modifications.  

\medskip
We proceed from the fact that states with fixed particle number
constitute superselection sectors of the algebra of
observables, \ie states with different particle numbers 
induce disjoint irreducible representations of this algebra.
Thus the particle number plays in the present context the role of a
charge quantum number. According to the deep 
insights of Doplicher, Haag and Roberts, such data
determine charge carrying morphisms of the observable
algebra, which connect different representations.  
These morphisms are given by the adjoint action
of isometric operators which can be interpreted as charge 
carrying fields. Different morphisms
carrying the same charge are related by observable  
intertwining operators. Such intertwiners also exist 
for so-called covariant morphisms which are
shifted by space
and time translations; they allow to extend these
shifts to the charged fields. So  
the observable algebra already contains 
the pertinent information about the underlying charged 
fields and their dynamics. 

\medskip
In the present case, the resolvent algebra is 
composed from the outset of particle number changing
operators. Nevertheless, it is meaningful to adopt the strategy
of Doplicher, Haag and Roberts, \ie to identify isometries whose 
adjoint actions define particle number changing morphisms of the 
observable subalgebra and to determine the corresponding intertwiners. 
In contrast to the relativistic case, these morphisms do not 
preserve the unit operator, and the intertwiners between 
them are only partial 
isometries. This is due to the fact that the
charged representations of the 
algebra of observables are not faithful in the present 
\mbox{non-relativistic} 
setting, \eg the vacuum representation
is only one dimensional. In spite of these differences, shifts of 
the morphisms by space and time translations can be defined 
since the algebra of observables is stable under these actions;
moreover, intertwining operators between the shifted morphisms 
exist. So one is faced with the question of whether these 
intertwiners are elements of the algebra of observables, \ie
whether the morphisms are covariant in the sense of Doplicher,
Haag and Roberts. The
proof that this is the case represents the technically 
most difficult part of the present investigation and
is deferred to the appendix. Having
settled this point, it follows 
that the C*-algebra, which is generated by the algebra
of observables and any one of the particle number
changing isometries, is stable under space and 
time translations. This algebra thus constitutes the desired 
extension of the observables to a 
bosonic field algebra which is stable under 
the action of symmetries and the dynamics. 

\medskip
Our article is organized as follows. In the subsequent
section, we collect 
some well known facts regarding canonical Bose fields,
establish our notation,  
and recall the definition of the resolvent algebra. 
In Sec.~3 we prove that this algebra admits a harmonic analysis 
with regard to the action of the gauge group.
We also recall some facts about the structure of the algebra 
generated by the gauge invariant observables 
and define the bosonic field algebra.
In Sec.~4 we consider localized morphisms of the observable
algebra and discuss their properties under symmetry
transformations. We present a condition in terms of
intertwining operators between these morphisms which
implies that these transformations can be extended to
automorphisms of the field algebra. 
The formalism is used in Sec.~5 for the discussion
of symmetries and dynamics. In particular, it is shown that the 
field algebra is stable under space and time translations for 
a large family of dynamics involving two-body potentials;
technical details are given in the appendix. 
The article closes with a summary and remarks 
on the treatment of further dynamics of 
physical interest. 
 
\section{Preliminaries}
\setcounter{equation}{0}

The resolvent algebra of canonical quantum systems 
has been abstractly defined in~\cite{BuGr1}. It is generated 
by symbols $R(\lambda, f)$, the resolvents of the underlying
canonical operators, where $\lambda \in \RR \backslash \{ 0 \}$
and, in the case of a scalar Bose field, $f \in \cD(\RR^s)$,
the space of complex valued test functions with 
compact support in position space. This space is regarded as a 
real symplectic space, equipped with some symplectic 
form $\sigma$, cf.\ below. 
The symbols $R(\lambda, f)$ satisfy a number of relations,  
encoding all algebraic properties of the fields, and their
polynomials generate a C*-algebra, the resolvent algebra. 
It was shown in~\cite{BuGr1}  that this algebra
is faithfully represented on the bosonic Fock space and we 
will deal with this concrete representation in the present
investigation. 

\medskip
We make use of the notation in \cite{Bu1} and  
denote by $\cF = \bigoplus_{n=0}^\infty \, \cF_n$ the symmetric 
(bosonic) Fock space. Its one-dimensional 
subspace $\cF_0$ consists of complex multiples of the vacuum 
vector $\bOmega$. The space $\cF_1 \simeq L^2(\RR^s)$ is the single
particle space with the standard scalar product
$\langle \Psi, \Phi \rangle \doteq \int \! d\bx \, 
\overline{\Psi}(\bx) \Phi(\bx)$,
and the $n$-particle subspace $\cF_n$ is spanned by the 
symmetric tensor products of single particle vectors,  
$ |\Phi_1 \rangle \otimes_s \cdots \otimes_s | \Phi_n \rangle$,  
$n \in \NN$.

\medskip
On Fock space $\cF$ there act 
the creation and annihilation operators 
$a^*$ and $a$, which are regularized 
with test functions $f,g \in \cD(\RR^s) \subset L^2(\RR^s)$.
They satisfy on their standard domains  
of definition the commutation relations 
$$
[a(f), a^*(g)] = \langle f, g \rangle \, \one \, ,
\quad [a(f), a(g)] = [a^*(f), a^*(g)] = 0 \, .
$$
We recall that $a^*(f)$ is complex linear in $f$ wheras $a(f)$, being
the hermitean conjugate of $a^*(f)$, is antilinear in $f$.

\medskip
This structure can be rephrased in terms of a single real linear,  
symmetric field operator~$\phi$ given by 
$\phi(f) \doteq \big(a^*(f) + a(f)\big)$, $f \in \cD(\RR^s)$. 
It satisfies the commutation relations
$$
[\phi(f) , \phi(g)] = i \sigma(f,g) \, \one
\, , \quad f,g \in \cD(\RR^s) \, ,
$$
where 
$\sigma(f,g) \doteq 2 \, \text{Im}\big(\langle f, g \rangle\big)$
is a non-degenerate real linear symplectic form on $\cD(\RR^s)$,
which thus is regarded as a symplectic space. Note that the
creation and annihilation operators can be recovered from the 
field by the formulas
$$ 
2 a^*(f) = \phi(f) - i \, \phi(if) \, , \quad
2 a(f) =  \phi(f) + i \, \phi(if) \, .
$$
The resolvents of the field operator,  
\begin{equation} \label{e2.1}
R(\lambda,f) \doteq \big(i \lambda  + \phi(f)\big)^{-1} \, , \quad 
\lambda  \in \RR \backslash \{ 0 \} \, , \ f \in \cD(\RR^s) \, , 
\end{equation}
generate, by taking their sums and products and
proceeding to the norm closure on $\cF$, the resolvent 
algebra $\bfR$, based on the symplectic space 
$(\cD(\RR^s), \sigma)$. As already mentioned, the 
algebra $\bfR$ provides a concrete and faithful
representation of the abstractly defined resolvent 
algebra, based on this symplectic space \cite[Thm.~4.10]{BuGr1}.

\section{Gauge transformations, tensors and observables}
\setcounter{equation}{0}

On the resolvent algebra $\bfR$ acts the global gauge group 
$U(1) \simeq \TT$
by maps $\gamma$, which are defined on the basic resolvents 
according to  
$$
\gamma_u\big(R(\lambda, f)\big) \doteq 
R(\lambda, e^{iu}f) \, , \quad u \in [0,2\pi] \, , 
$$
for $\lambda \in \RR \backslash \{ 0 \}$, $f \in \cD(\RR^s)$. 
These maps are unitarily implemented on Fock space 
by exponentials of the particle number operator $N$, 
$$ 
\gamma_u\big(R(\lambda, f)\big) = 
e^{iu N} R(\lambda, f) e^{-iu N} \, , \quad u \in [0,2\pi] \, .
$$
Thus these maps define a group of 
automorphisms $\gamma_{\, \TT}$ of the resolvent algebra.
This can also be seen in the abstract setting 
since the defining relations of the resolvent algebra remain unchanged 
under their action~\cite[Def.\ 3.1]{BuGr1}. 

\medskip
The action of the gauge group on the resolvent 
algebra $\bfR$ is \textit{not} pointwise norm continuous,
cf.\ \cite[Thm.\ 5.3(ii)]{BuGr1}. Nevertheless, one can 
perform a harmonic analysis of the elements of $\bfR$ 
with regard to this group 
by exploiting the fact that it acts pointwise 
continuously on $\bfR$ in the strong operator topology of 
the Fock representation. Wheras the definition of 
the harmonics by Fourier integrals relies on this weaker 
topology, the harmonics themselves are elements of the
C*-algebra $\bfR$, as is shown in the subsequent lemma.
Let us recall in this context that the resolvent algebra is faithfully 
represented on Fock space.

\begin{lemma} \label{l3.1}
Let $R \in \bfR$. The integrals 
$$
R_m \doteq (2 \pi)^{-1} \int_0^{2 \pi} \! du \, e^{-ium} \, e^{iuN} R e^{-iuN} \, ,
\quad m \in \ZZ \, ,
$$
being defined in the strong operator topology on $\cF$, are elements 
of the resolvent algebra, \ie $R_m \in \bfR$, $m \in \ZZ$. 
For fixed $m$, the operators 
$R_m$ transform as tensors (harmonics) under the 
gauge transformations,
$\gamma_u(R_m) = e^{ium} R_m$, $u \in [0,2\pi]$.
\end{lemma}

\noindent \textbf{Remark:} Note that there exist   
elements $R \in \bfR$, such as the basic
resolvents, which can not be 
approximated by the coresponding sums 
$\sum_m R_m$ in the norm topology. 
So harmonic synthesis fails in the C*-algebra $\bfR$. 
We will return to this point further below.

\begin{proof}
Since the polynomials of the basic resolvents are norm dense
in $\bfR$ and the map $R \mapsto R_m$ is norm
continuous, it suffices to establish the statement
for monomials. So, for $j = 1, \dots , k$, 
let \mbox{$\lambda_j \in \RR \backslash \{ 0 \}$}, 
$f_j \in \cD(\RR^s) \backslash \{ 0 \}$ and let 
$M \doteq \prod_{j=1}^k R(\lambda_j, f_j)$
be the corresponding ordered product of resolvents. Since
the function $u \mapsto e^{iuN}$ on $\cF$ is strong operator 
contiunuous, the integrals
$$
M_m = (2 \pi)^{-1} \int_0^{2 \pi} \! du \, e^{-ium} \, e^{iuN} M e^{-iuN} \, ,
\quad m \in \ZZ \, ,
$$
are defined in this topology. 
For the proof that they are elements of $\bfR$, let 
$L \subset L^2(\RR^s)$ be the complex subspace spanned 
by $f_1, \dots , f_k$ and let $\cF(L) \subset \cF$ 
be the Fock space based on  $L$. Since $(L, \sigma)$ 
is a finite dimensional non-degenerate 
symplectic subspace of $(\cD(\RR^s),\sigma)$, 
there is a corresponding resolvent algebra 
$\fR(L) \subset \bfR$ which is generated by the resolvents 
$R(\lambda,f)$, where \mbox{$\lambda \in \RR \backslash \{ 0 \}$}, 
\mbox{$f \in L$}.
This subalgebra acts faithfully on 
$\cF(L)$, cf.\ \cite[Thm.\ 4.10]{BuGr1}, 
and it contains some compact ideal which is 
represented on this by space by the algebra of compact operators,  
cf.\ \cite[Thm.\ 5.4]{BuGr1}.

\medskip
Now, given any $m \in \ZZ$, consider the function
$$
u,v \mapsto e^{i(u-v)m} e^{iuN} M^* e^{-iuN} e^{ivN} M e^{-ivN} \, , \quad
u,v \in [0, 2 \pi] \, . 
$$ 
Since 
$e^{iuN} R(\lambda, f) e^{-iuN} = R(\lambda, e^{iu} f)$,
and similarly for the adjoint resolvents, 
the values of this function lie in the intersections of the 
principal ideals in $\fR(L)$, which are generated
by the individual gauge-transformed resolvents in the  
above $2k$-fold product. According
to \cite[Prop.~4.4]{Bu3}, this intersection coincides 
with the principal ideal generated by the reordered product 
$$
R(\lambda_1, e^{iu} f_1)^*  R(\lambda_1, e^{iv} f_1)
\cdots R(\lambda_k, e^{iu} f_k)^*  R(\lambda_k, e^{iv} f_k) \, .
$$
Since the functions $f_1, \dots , f_k$ span the space 
$L$, the latter operator acts as a compact operator 
on $\cF(L)$ if all 
adjacent pairs of resolvents are generated by canonically 
conjugate operators~\cite[Thm.~5.4]{BuGr2}, \ie if 
$$
\sigma(e^{iu} f_j, e^{iv} f_j) = 
i (e^{i(v-u)} - e^{i(u-v)}) \, \langle f_j, f_j \rangle  \neq 0 
\quad \text{for} \quad j = 1, \dots, k \, .
$$ 
So the above function
has, for almost all $(u,v) \in [0, 2 \pi] \times [0, 2 \pi]$, 
values in compact operators on $\cF(L)$; moreover, it is bounded. 
Hence the double integral  
$$
M_m^{*} \, M_m = 
\int_0^{2 \pi} \! \! du \int_0^{2 \pi} \! \! dv \, 
e^{i(u-v)m} \, e^{iuN} M^* e^{-iuN} e^{ivN} M e^{-ivN}
$$
is a compact operator on this space as well. 
Taking its square root and performing a polar 
decomposition on $\cF(L)$, we find that 
$M_m \upharpoonright \cF(L)$ is also compact. It implies 
that $M_m$ is an element of the compact ideal of 
$\fR(L)$. Since $\fR(L) \subset \bfR$, we conclude that 
$M_m \in \bfR$ for any $m \in \ZZ$, completing
the proof of the main part of the statement. The remaining part
concerning the action of
the gauge transformations on the operators $R_m$ 
follows easily from the fact that on Fock space
this action can be interchanged with the integration.  \qed 
\end{proof} 

The preceding proposition augments a    
result in \cite{Bu1} according to which means over the
gauge group, corresponding to the value $m=0$ in
the above statement, map the algebra 
$\bfR$ into a subalgebra $\bfA \subset \bfR$ of gauge invariant 
operators. The latter operators preserve the particle numbers
of states and are interpreted as observables. 
It was crucial for the discussion of dynamics in \cite{Bu1} 
that one has detailed information about the structure of the 
algebra $\bfA$. Since this matters also 
in the present investigation we briefly recall here some
relevant facts. 

\medskip 
Given $n \in \NN_0$, the restriction
$\bfA \upharpoonright \cF_n$ defines an irreducible but
non-faithful representation of the observable algebra
on the $n$-particle space. The represented operators 
coincide with the elements of some C*-algebra $\fK_n$
on $\cF_n$, which has the following stucture: let $\fC_k$ be the 
algebra of compact operators on $\cF_k$, 
then the algebra $\fK_n$ is given by 
\begin{equation} \label{e3.1}
\fK_n \doteq \sum_{k = 0}^n \, \fC_k \otimes_s 
\underbrace{1 \otimes_s \cdots \otimes_s 1}_{n-k} \, , 
\end{equation}
where $\fC_0 \doteq \CC \, 1$. 
Note that the algebra of compact operators is
nuclar, so its C*-tensor products are unique and one also has 
for their symmetrized tensor products the relation 
\begin{equation} \label{e3.2}
\fC_k = \underbrace{\fC_1 \otimes_s \cdots \otimes_s \fC_1}_k \, .
\end{equation}
It has been shown in \cite[Lem.\ 3.3]{Bu1} that  
the restricted observable algebra satisfies 
the equality $\bfA \upharpoonright \cF_n = \fK_n$. 
Since these restrictions are not faithful, it is important to identify
in the algebras $\fK_n$, $n \in \NN_0$, those 
operators which arise from a given observable in $\bfA$. This
is accomplished by inverse maps $\kappa_n : \fK_n \rightarrow \fK_{n-1}$.
They are homomorphisms which act on the generating
operators of $\fK_n$ according to the formula,
$k = 0, \dots , n$, 
\begin{equation} \label{e3.3}
\kappa_n(C_k \otimes_s 
\underbrace{1 \otimes_s \cdots \otimes_s 1}_{n-k})
\doteq (n-k)/n \ C_k \otimes_s 
\underbrace{1 \otimes_s \cdots \otimes_s 1}_{n-k-1} \, ,
\quad C_k \in \fC_k \, ,
\end{equation}
cf.\ \cite{Bu1}. For notational convenience we 
put $\kappa_0 \doteq 0$ and define 
$\fK_n \doteq \{ 0 \}$ for $n < 0$.

\medskip
A sequence of operators $\bK \doteq \{ K_n \in \fK_n \}_{n \in \NN_0}$
is said to be coherent if its elements are 
uniformly bounded and $\kappa_n(K_n) = K_{n-1}$, $n \in \NN_0$.
Such coherent sequences are by definition 
the elements of the (bounded) 
inverse limit $\bfK$ of the inverse system 
$\{\fK_n, \kappa_n \}_{n \in \NN_0}$. This inverse limit is again a C*-algebra,
where the algebraic operations are component-wise defined. 
It has been shown in \mbox{\cite[Lem. 3.4]{Bu1}} that all elements 
$A \in \bfA$ of the algebra of observables determine such coherent 
sequences,
$\bK(A) \doteq \{ A \upharpoonright \cF_n \}_{n \in \NN_0} \in \bfK$,
but this map is not surjective. This can be remedied
by extending the algebra $\bfA$ to a C*-algebra $\obfA$ 
which consists of all bounded operators $\overline{A}$ on $\cF$
such that 
\begin{equation} \label{e3.4}
\overline{A} \upharpoonright {\textstyle \bigoplus_{k=0}^n} \, \cF_k \, \in \,  
\bfA \upharpoonright {\textstyle \bigoplus_{k=0}^n} \, \cF_k
\quad \text{for any} \ n \in \NN_0 \, .
\end{equation} 
So the restrictions of the 
algebras $\bfA$ and $\obfA$ coincide on all states
with limited particle number. It has been shown \cite[Thm.\ 3.5]{Bu1} 
that this extended algebra $\obfA$ is isomorphic to~the inverse 
limit $\bfK$, the isomorphism being given by the map 
\begin{equation} \label{e3.5}
\overline{A} \mapsto 
\bK(\overline{A}) \doteq 
\{ \overline{A} \upharpoonright \cF_n \}_{n \in \NN_0} \in \bfK \, ,
\quad \overline{A} \in \obfA \, .
\end{equation}

\medskip
We turn now to the operators in $\bfR$ which transform as 
arbitrary tensors under gauge transformations. As already
mentioned, these operators do not generate the full
resolvent algebra (harmonic synthesis fails).
We therefore proceed to a more convenient algebra, which admits
harmonic analysis and synthesis and also 
contains the extended algebra of observables $\obfA$.

\medskip 
Let $\bfR_M$, $M \in \NN_0$, be the (norm closed) subspace of $\bfR$ which is
generated by all linear combinations of tensors
$R_m \in \bfR$ with $-M \leq m \leq M$.
Clearly, \mbox{$\bfR_{M_1} \subset \bfR_{M_2}$} if 
$M_1 \leq M_2$. Since any tensor satisfies the equality
$R_m{}^{\! *} = R^{\, *}{}_{\! -m}$, the
spaces $\bfR_M$ are symmetric, $\bfR_M{}^{\! *} = \bfR_M$.
It is also clear that the product
of tensors is again a tensor,
$R^\prime_{m^\prime} \, R^{\prime \prime}_{m^{\prime \prime}}
= R_{m^\prime + m^{\prime \prime}}$; hence one has 
with regard to pointwise multiplication the inclusion 
$\bfR_{M^\prime} \, \bfR_{M^{\prime \prime}} \subset
\bfR_{M^\prime + M^{\prime \prime}}$.  
It follows that $\bigcup_M \bfR_M \subset \bfR$ is a
*-algebra whose
norm closure in $\bfR$ will be denoted by $\bfF$ 
and called field algebra. It contains 
the algebra of observables,~$\bfA = \bfR_0 \subset \bfF$. 

\medskip 
In order to incorporate also the extended algebra $\obfA$, we proceed 
similarly as in case of the observables and consider the 
set of bounded operators $\overline{F}$ on~$\cF$
for which there is some $M \in \NN_0$ such that 
\begin{equation} \label{e3.6}
\overline{F} \upharpoonright {\textstyle \bigoplus_{k=0}^n} \,\cF_k
\in \bfR_M \upharpoonright {\textstyle \bigoplus_{k=0}^n} \, \cF_k
\quad \text{for all} \ n \in \NN_0 \, .
\end{equation}
The closure of this space of operators with regard to the 
operator norm on~$\cF$ is denoted by~$\obfF$. We show
in the subsequent lemma that $\obfF$ is a C*-algebra, which has 
the desired properties. Since it contains the algebra $\bfF$,
we will refer to it as extended field algebra. 

\begin{lemma} \label{l3.2}
The norm closed space $\obfF$, defined above, is a C*-algebra.
It extends the field algebra, $\obfF \supset \bfF$, and
contains the extended algebra of observables, 
\mbox{$\obfA \subset \obfF$}. The gauge group acts pointwise 
norm continuously on $\obfF$ by the adjoint action of 
the exponentials of the number operator, and harmonic analysis 
and synthesis do work on this algebra. In particular,
$\obfA$ is the fixed point algebra in~$\obfF$ under gauge
transformations. 
\end{lemma}
\begin{proof}
For the proof that the space $\obfF$ is a C*-algebra, it suffices to 
show that it is stable under taking products and adjoints
since, by definition, it is closed with regard to the operator 
norm on $\cF$. Moreover, it is sufficient to establish these
features for the norm-dense subspace of operators satisfying
condition \eqref{e3.6}. 

\medskip
Turning to the products, note that for
any tensor $R_m$ one has $R_m \cF_n \subset \cF_{m+n}$ if
$m+n \geq 0$ and $R_m \cF_n = 0$ if $m+n < 0$.
Hence \ $\bfR_M  \bigoplus_{k=0}^n \cF_k
\subset \bigoplus_{k=0}^{n + M} \cF_k$ for $n, M \in \NN_0$.  
Now let $\overline{F}_1, \overline{F}_2 \in \obfF$ be operators
satisfying condition \eqref{e3.6} and pick any $n \in \NN_0$. 
Then there is some $M_1 \in \NN_0$ which does not depend on 
$n$ and some operator $R_{1,n} \in \bfR_{M_1}$ such that 
$$
\overline{F}_1 \upharpoonright \textstyle{\bigoplus_{k=0}^n} \, \cF_k
=  R_{1,n} \upharpoonright \textstyle{\bigoplus_{k=0}^n} \, \cF_k
\subset \, \textstyle{\bigoplus_{k=0}^{n + M_1}} \, \cF_k \, .
$$
Similarly, there is some $M_2 \in \NN_0$ which does not depend on $n$ and 
some operator $R_{2,n + M_1} \in \bfR_{M_2}$ such that 
$$
\overline{F}_2 \upharpoonright \textstyle{\bigoplus_{k=0}^{n + M_1}} \, \cF_k
=  R_{2,n + M_1} \upharpoonright \textstyle{\bigoplus_{k=0}^{n + M_1}} \, \cF_k
\, .
$$
Since $R_{2,n + M_1} R_{1,n} \in \bfR_{M_2}  \bfR_{M_1} \subset 
 \bfR_{M_1 + M_2}$, it follows that 
$$
\overline{F}_2 \overline{F}_1 \upharpoonright 
\textstyle{\bigoplus_{k=0}^n} \, \cF_k = 
R_{2,n + M_1} R_{1,n} \upharpoonright 
\textstyle{\bigoplus_{k=0}^n} \, \cF_k 
\in \bfR_{M_1 + M_2} \upharpoonright 
\textstyle{\bigoplus_{k=0}^n} \, \cF_k \, ,
$$
proving that the product $\overline{F}_2 \overline{F}_1$ also satisfies
condition \eqref{e3.6}.

\medskip 
Next, let $\overline{F}$ be an operator 
satisfying condition \eqref{e3.6}
for some $M \in \NN_0$. Picking any $k \geq M$, it follows that 
$\overline{F} \, \cF_k \subset \textstyle{\bigoplus_{l=k-M}^{k + M}} \, \cF_l$.
Hence the space $\overline{F} \, \cF_k$ is orthogonal to 
$\bigoplus_{l=0}^n \, \cF_l$ for $k > n + M$. 
It implies that the adjoint
of $\overline{F}$ satisfies $\overline{F}^{\, *} \bigoplus_{k=0}^n \, \cF_k
\subset \bigoplus_{k=0}^{n + M} \, \cF_k$ for any $n \in \NN_0$.  
Now let $R_{n + M} \in \bfR_M$ be an operator such that 
$\overline{F} \upharpoonright \bigoplus_{k=0}^{n + M} \, \cF_k
= R_{n + M} \upharpoonright \bigoplus_{k=0}^{n + M} \, \cF_k$. 
Picking arbitrary vectors $\bPhi_n \in \bigoplus_{k=0}^{n} \, \cF_k$
and $\bPsi_{n + M} \in \bigoplus_{k=0}^{n + M} \, \cF_k$, one obtains
$$
\langle \bPsi_{n + M}, \overline{F}^* \bPhi_n \rangle =
\langle \overline{F} \, \bPsi_{n + M},  \bPhi_n \rangle =
\langle  R_{n + M} \, \bPsi_{n + M},  \bPhi_n \rangle =
\langle \bPsi_{n + M}, {R_{n + M}}^{\! *} \, \bPhi_n \rangle \, .
$$
Hence $ \overline{F}^* \! \upharpoonright \bigoplus_{k=0}^{n} \, \cF_k
= {R_{n + M}}^{\! *} \upharpoonright \bigoplus_{k=0}^{n} \, \cF_k$.
Since ${R_{n + M}}^{\! *} \in \bfR_M$, we conclude that 
$\overline{F}^*$ satisfies condition \eqref{e3.6},
showing that $\obfF$ is a C*-algebra. 

\medskip 
That $\obfF$ extends the field algebra $\bfF$ follows from the
fact that the elements of the *-algebra $\bigcup_M \bfR_M$
satisfy condition \eqref{e3.6} and are therefore contained in 
$\obfF$. But $\bfF$ is by definition the norm closure of 
this *-algebra and hence is contained in the norm closed
algebra $\obfF$. It is also clear that $\obfF$ contains
the extended observable algebra $\obfA$ which is obtained 
by restricting condition \eqref{e3.6} to operators,
which satisfy this condition for $M=0$.

\medskip 
It remains to establish the statements concerning the action
of the gauge transformations, where it suffices again to verify
them for all operators satisfying condition \eqref{e3.6}. So let
$\overline{F}$ be any such operator which satisfies this 
condition for some $M \in \NN_0$. Given $n \in \NN_0$, there is
some operator $R_n \in \bfR_M$ such that 
$\overline{F} \upharpoonright \bigoplus_{k = 0 }^n \cF_k
= R_n \upharpoonright \bigoplus_{k = 0 }^n \cF_k $. 
Since the space $\bigoplus_{k = 0 }^n \cF_k$ is invariant 
under the action of the unitaries 
$e^{-iuN}$, $u \in [0,2\pi]$, one obtains for any $m \in \ZZ$
the equality of integrals, defined in the 
strong operator topology,
$$
\int_0^{2 \pi} \! \! \! du \, e^{-ium} \, e^{iuN} \overline{F} e^{-iuN}
\upharpoonright {\textstyle \bigoplus_{k = 0 }^n} \, \cF_k =
\int_0^{2 \pi} \! \! \! du \, e^{-ium} \, e^{iuN} R_n e^{-iuN}
\upharpoonright {\textstyle  \bigoplus_{k = 0 }^n} \, \cF_k \, .
$$
It follows from the definition of $\bfR_M$ that
$\int_0^{2 \pi} \! du \, e^{-ium} \, e^{iuN} R_n e^{-iuN}
\in \bfR_M$ for $|m| \leq M$ and that the integrals
vanish if $|m| > M$. Consequently, the bounded operators
$\overline{F}_m \doteq (2 \pi)^{-1}
\int_0^{2 \pi} \! du \, e^{-ium} \, e^{iuN} \overline{F} e^{-iuN}$,
$m \in \ZZ$, satisfy condition \eqref{e3.6}. Hence they are
elements of $\obfF$ and harmonic analysis
is possible for them.  Moreover,
$\overline{F} = \sum_{m = -M}^M \overline{F}_m$,
\ie harmonic synthesis holds as well. It is also
apparent from the latter equality that the gauge
transformations act norm continuously on these
operators, so the above integrals are even defined in
the norm topology. Since these special
operators form a norm dense subset of $\obfF$
and the maps $\overline{F} \mapsto \overline{F}_m$
are norm continuous, $m \in \ZZ$, the preceding 
properties are shared by all elements of $\obfF$.

\medskip Finally, let $\overline{F} \in \obfF$ be any gauge 
invariant operator. According to the definition of 
$\obfF$ it can be approximated in norm by operators 
satisfying condition \eqref{e3.6}. Taking a mean over the
gauge group one sees that $\overline{F}$ can also be 
approximated in norm by operators satisfying this condition 
for $M = 0$, \ie by elements of $\obfA$. But 
$\obfA$ is norm closed, so $\overline{F} \in \obfA$,
completing the proof.
\qed \end{proof} 

Having clarified the properties of the extended field 
algebra $\obfF$, we will exhibt now its structure in more concrete terms. 
We will make use of the fact that $\obfF$
contains isometries which
transform as elementary tensors under the action of the gauge group. 
Let $f \in \cD(\RR^s) \subset L^2(\RR^s)$ be normalized, $\| f \|_2 = 1$,
and let $N_f \doteq a^*(f)a(f)$. We define on $\cF$ the operators
\begin{equation} \label{e3.7}
W_f \doteq a(f)^* (1 + N_f)^{-1/2} \, , \quad 
W_f^{\, *} = (1 + N_f)^{-1/2} a(f) \, .  
\end{equation}
It follows from the commutation relations of
the creation and annihilation operators 
that $W_f^{\, *} W_f = 1$. Thus $W_f$ is an isometry 
and $W_f W_f^{\, *} = E_f$ is the projection 
onto the orthogonal complement of the kernel of $a(f)$ in $\cF$. Moreover, 
\begin{equation} \label{e3.8}
e^{iuN} W_f e^{-iuN} = e^{iu} \, W_f \,, \quad  u \in [0,2 \pi] \, , 
\end{equation}
and an analogous relation holds for the adjoint operator $W_f^*$. 

\medskip
It is essential for the subsequent analysis that 
$W_f, \, W_f^* \in \obfF$. 
For the proof of this statement, 
we proceed as in Lemma~\ref{l3.1}. 
Let $L \subset L^2(\RR^s)$ be the complex ray
spanned by $f$, let $\cF(L) \subset \cF$ be the
Fock space based on $L$, and let $\bfR(L) \subset \bfR$
be the subalgebra of the resolvent algebra, which is
generated by resolvents with test functions
in $L$. The operator $(1 + N_f)^{-1/2}$ is gauge 
invariant and acts as a compact operator on $\cF(L)$, 
so the spectral projections $P_f(k)$, $k \in \NN_0$, in its 
spectral resolution 
have finite rank on $\cF(L)$. Consequently, the operators 
$(1 + N_f)^{-1/2} a(f) \, \big(\sum_{k=0}^n P_f(k)\big)$ also have 
finite rank on~$\cF(L)$. They are therefore elements of the compact ideal 
of $\bfR(L)$, $n \in \NN_0$. Moreover, these operators are 
tensors under gauge transformations 
corresponding to the fixed value $m = -1$. Since $P_f(k) \, \cF_n = 0$ if
$k > n$, this implies that for any~$n \in \NN_0$
$$
W_f^* \upharpoonright {\textstyle \bigoplus_{k = 0}^n} \, \cF_k
=  (1 + N_f)^{-1/2} a(f){\textstyle  \, \big(\sum_{k=0}^n P_f(k)\big) 
\upharpoonright \bigoplus_{k = 0}^n} \, \cF_k
\subset \bfR_1 \upharpoonright {\textstyle \bigoplus_{k = 0}^n} \, \cF_k \, .
$$
Thus $W_f^*$ satisfies condition \eqref{e3.6} for $M = 1$,
and the same argument applies to~$W_f$, proving the 
assertion. It is easy now
to establish the following proposition,
which provides the basis for the
subsequent discussions. 

\begin{proposition} \label{p3.3}
  Let $f \in \cD(\RR^s)$ be a fixed, normalized test function,
  $\| f \|_2 =1$. The 
extended field algebra $\obfF$ coincides with the C*-algebra that 
is generated by the extended algebra of observables $\obfA$
and the isometric tensor $W_f, W_f^* \in \obfF$,
defined in Eqn.~\eqref{e3.7}. 
\end{proposition}
\begin{proof}
By definition, $\obfF$ is the C*-algebra generated by the
space of all bounded operators on $\cF$ that
satisfy condition \eqref{e3.6} for some $M \in \NN_0$. 
The statement of Lemma~\ref{l3.2},  
regarding harmonic analysis and synthesis,
implies that this space coincides with the span
of tensor operators $\overline{F}_m \in \obfF$, $m \in \ZZ$. 
Given any such tensor for the value $m = 0$, one has
$\overline{F}_0 \in \obfA$ according to Lemma~\ref{l3.2}. 
If $\overline{F}_m \in \obfF$ corresponds to some value $m > 0$, 
the operator
$\overline{A}_m \doteq \overline{F}_m \, W_f^{* m}  \in \obfF$ is gauge
invariant, cf.\ Eqn.~\eqref{e3.8}, so $\overline{A}_m \in \obfA$. 
But $\overline{F}_m  = \overline{F}_m \, W_f^{* m} \, W_f^m =
\overline{A}_m \, W_f^m$, where we made use of the fact that
$W_f$ is an isometry. This shows that $\overline{F}_m$
is an element of the algebra generated by $\obfA$ and $W_f, W_f^*$.
Similarly, if $\overline{F}_m \in \obfF$ corresponds to some value $m < 0$, 
then $\overline{A}_m \doteq W_f^m \, \overline{F}_m \in \obfA$, so  
$\overline{F}_m = W_f^{* m} \, W_f^m \, \overline{F}_m = W_f^{* m} \,
\overline{A}_m$ is also contained in the algebra
generated by  $\obfA$ and $W_f, W_f^*$. The statement then follows. 
\qed \end{proof}  

\medskip
The algebra of observables $\obfA$ and the field algebra $\obfF$
were concretely constructed on Fock space, which we regard as
their defining (hence faithful) representations. 
Since we will deal in the following exclusively 
with these extended algebras, we
omit the term ``extended'' from now on, speaking
simply of observables and fields. To simplify the
notation, we also omit the bar
${}^{\mbox{\bf --}}$ from the elements of these algebras.

\section{Morphisms, intertwiners and covariance}
\setcounter{equation}{0}

In this section we adjust 
the framework of Doplicher, Haag and Roberts, mentioned in the
introduction, to the non-relativistic case. 
For the convenience of the reader, 
we recall in some detail the basic notions underlying their 
approach and point out some differences with regard to the present 
situation. After a recapitulation 
of the local structure of the 
algebra of observables $\obfA$ and properties of its  
representations, we recall the concept of localized morphisms 
of $\obfA$ and explain how it is related to the field algebra $\obfF$. 
We will see that the question of invariance of the
field algebra $\obfF$ under the action of 
symmetry transformations, such as the time
evolution for given dynamics, is equivalent to
the question of covariance of the localized morphisms.
Thus the answer to this question is  
encoded in the observable algebra. This insight justifies 
our two step approach to the problem, where we first focussed  
on the observables \cite{Bu1} and now turn to the fields. 

\medskip 
As already pointed out in \cite{Bu1}, the algebra $\obfA$ has some 
local structure with regard to the underlying space $\RR^s$.   
Given any bounded region $\bO \subset \RR^s$ with open
interiour, one first defines
a corresponding algebra $\bfR(\bO) \subset \bfR$: 
it is the C*-algebra generated by all resolvents 
$R(\lambda,f)$, where $\supp f \subset \bO$, 
$\lambda \in \RR \backslash \{ 0 \}$. Its gauge invariant
subalgebra is denoted by $\bfA(\bO)$. The corresponding 
extended algebra $\obfA(\bO)$ is then defined as the set of bounded
operators on $\cF$ which satisfy the strengthened condition \eqref{e3.4},
where $\bfA$ is replaced by $\bfA(\bO)$. One has by
construction $\, \obfA(\bO_1) \subset \obfA(\bO_2)$ if $\, \bO_1 \subset \bO_2$,
and the algebras corresponding to disjoint sets commute as
a consequence of the canonical commutation relations,
$[ \obfA(\bO_1), \obfA(\bO_2)] = 0$ if   
$\bO_1 \bigcap \bO_2 = \emptyset$. 
So the assignment $\bO \mapsto \obfA(\bO)$ defines 
a local net of C*-algebras on~$\RR^s$. 
However, in contrast to the relativistic case,
the C*-inductive limit 
of this net does not coincide with the global algebra $\obfA$. 
But it is still true that the restriction of 
$\bigcup_{\sbO \subset \RR^s} \obfA(\bO)$
to any subspace of $\cF$ with limited particle number is norm dense 
in the restriction of $\obfA$ to the respective space.  

\medskip 
The basic irreducible representations of the algebra $\obfA$
are obtained by restricting it to the subspaces $\cF_n \subset \cF$
for given particle number $n \in \NN_0$. As was already mentioned,
these representations are disjoint for different values of~$n$,
\ie the corresponding states are superselected. 
We briefly indicate the proof: let $L \subset L^2(\RR^s)$
be any finite dimensional complex subspace
and let $N_L$ be the particle number 
operator on $\cF_L \subset \cF$.
The resolvents $(1 + N_L)^{-1}$ are elements 
of~$\bfA$, cf.\ the proof of Lemma~\ref{l3.1}. Now for any 
increasing sequence of subspaces $L$, exhausting $L^2(\RR^s)$,
the limit of the corresponding resolvents exists in the strong operator 
topology on $\cF$ and is given by $(1 + N)^{-1}$,
the resolvent of the number operator. It is contained in 
the closure of $\bfA$ with regard to this topology,
commutes with all elements of $\bfA$ and has different 
sharp values on the 
subspaces $\cF_n$, $n \in \NN_0$. So this limit defines a 
superselected global observable that  
distinguishes the corresponding irreducible 
representations of~$\obfA$. 

\medskip
In close analogy to the relativistic case, there exist
also in the present situation  
states in any of these representations which cannot
be discriminated from states in any other representation 
by observables which are localized in the complement~$\bO_0^\perp$ of 
a given region $\bO_0 \subset \RR^s$. In particular, 
given $n \in \NN_0$,  there exist
vectors $\bPhi_n \in \cF_n$ such that for any 
$\bO \subset \bO_0^\perp$ one has 
$$
\langle \bPhi_n, A \bPhi_n \rangle = \langle \bOmega, A 
\bOmega \rangle \, , \quad  A \in \obfA(\bO) \, .
$$
This relation agrees with the selection criterion 
of Doplicher, Haag and Roberts in relativistic quantum field
theory for states that can be interpreted as local excitations of the
vacuum. As they have shown, these  
states can be obtained by composition of the 
vacuum state with morphisms of the algebra of observables,
where a morphism is a linear, symmetric and
multiplicative map of this algebra into itself. 
Their conclusion fails, however, in the present setting  
since the vacuum representation of $\obfA$ is one-dimensional.
In contrast, all local observables in relativistic 
quantum field theory retain their quantum
nature in the vacuum representation due to fluctuations 
caused by neutral 
particle-antiparticle pairs. This feature is excluded 
from the outset in the non-relativistic setting since the 
generator $N$ of the gauge transformations is bounded from below.

\medskip 
Wheras there do not exist morphisms of $\obfA$ whose
composition with vector states in $\cF$ increases the particle number,
there can and (as we shall see) do exist morphisms $\rho$ which decrease 
this number. More concretely, for any $n \in \NN_0$ and $\bPhi_n \in \cF_n$,
there exists some vector $\bPhi_{n-1} \in \cF_{n-1}$ such that 
$$
\langle \bPhi_n, \rho(A) \, \bPhi_n \rangle =
\langle \bPhi_{n-1}, A \, \bPhi_{n -1} \rangle \, , \quad
A \in \obfA \, ,
$$
where we put $\cF_n \doteq \{ 0 \}$ if $n < 0$. 
One then has $\langle \bOmega, \rho(\one) \bOmega \rangle = 0$,
\ie $\rho(\one)$ is a projection in $\obfA$ 
which is different from $\one$ and 
has the vacuum in its kernel. 
Any such morphism defines a representation 
$\rho : \obfA \rightarrow \rho(\obfA)$
on Fock space~$\cF$. Since $\rho(\obfA)$ annihilates 
the space $\big(\one - \rho(\one)\big) \cF$, it is sensible to restrict 
this representation to $\rho(\one) \cF \subset \cF$,
on which $\rho(\obfA)$ is faithfully represented. 
Taking this into account, it is apparent 
how to extend the concept of localized morphism,  
used in the relativistic setting, to the present framework: the  
morphism $\rho$ is said to be localized in 
a given region $\bO_0 \subset \RR^s$ if it satisfies for any 
contractible region $\bO \subset \bO_0^\perp$  
$$
\rho(A) = A \, \rho(\one) = \rho(\one) \, A \, , \quad A \in \obfA(\bO) \, .
$$
One then has for any $\bPhi_n \in \rho(\one)  \cF_n$, $n \in \NN_0$, 
and $\bO \subset \bO_0^\perp$ 
$$ \langle \bPhi_n, \rho(A) \bPhi_n \rangle 
= \langle \bPhi_n, A \bPhi_n \rangle \, ,
\quad A \in \obfA(\bO) \, .
$$ 
Comparing this equality with the relation given above, 
one sees that states in different sectors cannot be
distinguished by observations in the
complement of the localization region of $\rho$, 
similarly to the relativistic case.  

\medskip 
In the subsequent analysis we will have to compare morphisms
$\rho_1, \rho_2$ of the algebra $\obfA$ which differ,  
for example, by their localization or by some symmetry transformation. 
Since observables
do not change the particle number, it is physically
meaningful to form classes of morphisms which can be transformed into
each other by operations described by elements of $\obfA$. 
We say the morphisms are 
equivalent and write $\rho_1 \simeq \rho_2$ if there exists some
partial isometry $X_{1,2} \in \obfA$ such that 
$X_{1,2} X_{1,2}^* = \rho_1(\one)$, $X_{1,2}^* X_{1,2} = \rho_2(\one)$
and $\rho_1(A) X_{1,2} = X_{1,2} \, \rho_2(A)$, $A \in \obfA$. 
Such partial isometries are called intertwining operators. 
Their existence implies that the representation $\rho_1$ of $\obfA$,
which is defined on $\rho_1(\one) \cF$, is equivalent to 
the representation $\rho_2$, which is defined on 
$\rho_2(\one) \cF$. So $\rho_1$ and $\rho_2$
describe physically indistinguishable representations of 
the observables. 

\medskip 
After this outline of concepts related to the observable
algebra and its basic representations, let us discuss now  
how they can be realized in terms of the field
algebra $\obfF$. Doplicher and 
Roberts have shown how the field algebra and the gauge group 
can be recovered in the relativistic setting from localized 
morphisms of the observables \cite{DoRo}. 
We are here in a more comfortable situation since the
algebra~$\obfF$ is already at our disposal. Yet its 
specific definition was guided by the above background
information. 

\medskip 
For the construction of localized morphisms of $\obfA$
we make use of the isometries,
defined in Eqn.~\eqref{e3.7}, which generate together with the
observables the field algebra $\obfF$.  We pick a fixed
$f \in \cD(\RR^s)$ that is normalized and consider the map \
$\rho_f : \obfA \rightarrow \rho_f(\obfA)$ \ given by
\begin{equation} \label{e4.1}
\rho_f(A) \doteq W_f \, A \, W_f^* \, , \quad A \in \obfA \, .
\end{equation}
It follows from results in the preceding section that the 
range of this map lies in~$\obfA$: \ $W_f \, A \, W_f^*$ is 
according to relation \eqref{e3.8} and 
Proposition~\ref{p3.3} a gauge invariant element 
of~$\obfF$ and thus contained in $\obfA$, cf.\ 
Lemma~\ref{l3.2}. Moreover, $\rho_f$ is linear,
symmetric, and multiplicative (since $W_f$ is an isometry).
Hence $\rho_f$ is a morphism of $\obfA$ and one has 
$\rho_f(\one) = W_f W_f^* = E_f$. 
It is also easily seen that $\rho_f$ is localized
in the region $\supp f \subset \RR^s$. Finally, 
for any two normalized elements $f_1, f_2 \in \cD(\RR^s)$ 
the corresponding morphisms $\rho_{f_1}$ and $\rho_{f_2}$ are
equivalent: they are related by the partial isometry
$X_{f_1, f_2} \doteq W_{f_1}  W_{f_2}^*$, which is a gauge invariant element of 
$\obfF$ and therefore contained in $\obfA$. 
Thus the choice of localization region of the morphisms 
is just a matter or convenience, they are 
transportable in the sense of Doplicher, Haag and Roberts. 
As a matter of fact, all morphisms satisfying the 
preceding conditions are related to the morphisms
$\rho_f$ by intertwining operators in  
the weak closure of some local subalgebra of~$\obfA$. 

\medskip
Let us turn now to the central issue of this section, namely
the question under which circumstances the action of a 
symmetry group on the observable algebra $\obfA$ can be extended to the 
field algebra $\obfF$. So let $G$ be some group that is 
represented on $\cF$ by unitary operators, $g \mapsto U(g)$;
its adjoint action, denoted by $\ad U(g)$,   
is assumed to leave the algebra of observables
invariant, $\ad U(g) \, (\obfA) = \obfA$,  $g \in G$.
Note that this relation does not fix the unitary operators 
since the algebra $\obfA$ has a non-trivial commutant 
on $\cF$ that consists of the abelian von Neumann algebra~$\bfN$ 
generated by the particle number operator $N$. (The latter 
assertion follows from the fact 
that $\bfN$ is contained in the weak closure of
$\obfA$ and that the restrictions 
$\obfA \upharpoonright \cF_n$ are irreducible, $n \in \NN_0$.)
We will restrict our attention to the cases where 
the unitaries $U(g)$, $g \in G$, commute with $N$.
This covers the space-time
translations which are of primary interest  here. 

\medskip 
The action of $G$ on $\obfA$ can be lifted to the morphisms $\rho_f$, 
putting 
\begin{equation} \label{e4.2}
{}^g \! \rho_f \doteq \ad U(g) \ \scirc \, \rho_f \, 
\scirc \, \ad U(g)^* \, ,
\quad g \in G \, ,
\end{equation}
where the circle indicates the composition of maps. 
Clearly, the maps ${}^g \! \rho_f$ are again morphisms 
of $\obfA$, $g \in G$. Moreover, one has 
${}^{g_1} ({}^{g_2} \! \rho_f) = {}^{g_1 g_2} \! \rho_f$ for
$g_1, g_2 \in G$.
Note that we do not require that the transformed morphisms   
are also localized, which would be too restrictive 
an assumption in the present non-relativistic setting.  In analogy to 
Doplicher, Haag and Roberts, we say that the
morphism $\rho_f$ transforms
covariantly under the action of $G$ if 
${}^g \! \rho_f \simeq \rho_f$ for all $g \in G$.
The following lemma shows that the problem of stability
of the field algebra $\obfF$ with regard to the action of the group 
$G$ is encoded in properties of the observables.

\begin{lemma} \label{l4.1}
Let $G$ be a group and let $U$ be a unitary representation 
of $G$ on~$\cF$ which commutes with the particle number
operator $N$ and whose adjoint action leaves the observable algebra
$\obfA$ invariant, $\ad U(g)(\obfA)  = \obfA$, 
$g \in G$. The following statements are equivalent. 

\medskip
\noindent (i) \ $\rho_f$ is covariant with regard
to the action of $G$,  \ ${}^g \! \rho_f \simeq \rho_f$, \ $g \in G$.  

\medskip
\noindent (ii) There exists a unitary representation 
$g \mapsto Z_{f,N}(g)$ 
with values in the 
von Neumann algebra~$\bfN$, generated by $N$, such that \
$W_f \ad U(g)(W_f^*)  \, Z_{f,N}(g)^* \in \obfA$, \ $g \in G$.

\medskip
\noindent (iii)  There exists a unitary representation 
$g \mapsto U_N(g)$ with values in $\bfN$ such that 
the representation $g \mapsto V(g) \doteq U(g) U_{N}(g)$ satisfies 
$\ad V(g)(\obfF) = \obfF$, \mbox{$g \in G$}.
Thus the representation $V$ of $G$ extends the adjoint action of $U$ 
from the observable algebra to the field algebra.
\end{lemma}

\vspace*{-2mm}
\begin{proof}
(i) $\rightarrow$ (ii): \ Let $Y_f(g) \doteq W_f \, \ad U(g)(W_f^*)$, 
$g \in G$. 
Then 
$$ Y_f(g) Y_f(g)^* = E_f = \rho_f(\one) \, , \quad
Y_f(g)^* Y_f(g) = \ad U(g)(E_f) = {}^g \! \rho_f(\one)
$$ 
and, $A \in \obfA$,   
$$
\rho(A) \, Y_f(g) = 
W_f \underbrace{\ad U(g)(W_f^*) \ad U(g)(W_f)}_1A  \, 
\ad U(g)(W_f^*) = 
Y_f(g) \, {}^g \! \rho_f(A) \, .
$$ 
Thus $Y_f(g)$ intertwines the morphisms $\rho_f$ and ${}^g \! \rho_f$.
Because of the ambiguities involved in the choice of the representation
$U$ of $G$, these operators may not be contained in $\obfA$,
however.
To deal with this problem, we proceed from the intertwiners $X_f(g) \in \obfA$ 
between $\rho_f$ and ${}^g \! \rho_f$, which exist according to (i),
and compare them with the operators $Y_f(g)$, defined above. 
Let $Z_f(g) \doteq Y_f(g) X_f(g)^*$, $g \in G$.
These operators satisfy 
$$
\rho_f(A) \, Z_f(g) = Y_f(g) \, {}^g \! \rho_f(A) \, X_f(g)^* =
Z_f(g) \, \rho_f(A) \, , \quad A \in \obfA \, ,
$$
and one easily verifies that 
$Z_f(g) Z_f(g)^* = Z_f(g)^* Z_f(g) = \rho_f(\one)$. 
Since the algebra $\rho_f(\obfA)$ acts irreducibly on the subspaces
$\rho_f(\one) \, \cF_n$, it follows from the preceding
relations that the restrictions of $Z_f(g)$ to these subspaces
are given by phase factors $\zeta_{f, n}(g) \in \TT$, $n \in \NN$.
Choosing any 
$\zeta_{f,0}(g) \in \TT$, which will be fixed
below, we define corresponding unitary operators
$Z_{f, N}(g)$ on $\cF$ by the equations 
$Z_{f, N}(g) \upharpoonright \cF_n \doteq \zeta_{f,n}(g) 1 \upharpoonright \cF_n$,
$n \in \NN_0$. These operators
are elements of $\bfN$ and one has 
$Z_f(g) = Z_{f, N}(g) \, \rho_f(\one) $, \ $g \in G$. 
Multiplying this 
relation from the right by $X_f(g)$ and taking into account that 
$\rho_f(\one) X_f(g) = X_f(g)$, we arrive at the equality  
$Y_f(g) = Z_{f,N}(g) X_f(g)$. 

\medskip
It follows from this equality that 
$g \mapsto Z_{f,N}(g)$ can be lifted to a true
representation of $G$ by fixing the values of 
$\zeta_{f,0}(g)$ and multiplying the resulting 
operators with $\overline{\zeta_{f,1}}(g)$;
the compensating factor $\zeta_{f,1}(g)$
can be absorbed in the observables $X_f(g)$.
To verify this we make use of the fact
that $g \mapsto U(g)$ defines a representation of $G$. The 
above equality implies that for $g_1,g_2 \in G$
\begin{align*}
\rho_f(\one) & = Y_f(g_1) Y_f(g_2) Y_f(g_1 g_2)^* \\ & = 
X_f(g_1) X_f(g_2) X_f(g_1 g_2)^* \, Z_{f,N}(g_1) Z_{f,N}(g_2) Z_{f,N}(g_1 g_2)^* \, .
\end{align*}
Since the operators $Z_{f,N}(g)$ are unitary, we can proceed to 
$$
\rho_f(\one) \, Z_{f,N}(g_1 g_2) Z_{f,N}(g_1)^* Z_{f,N}(g_2)^*
= X_f(g_1) X_f(g_2) X_f(g_1 g_2)^* \, .
$$
Multiplying this equality from the left by $W_f^*$,  from the right
by $W_f$, and taking into account that 
$W_f^* N W_f = N  + \one$, we arrive at
$$
 Z_{f,N+1}(g_1 g_2) \, Z_{f,N+1}(g_1)^* \, Z_{f,N+1}(g_2)^*
= W_f^* X_f(g_1) X_f(g_2) X_f(g_1 g_2)^* \, W_f \, , 
$$
where the operators $Z_{f,N+1}(g)$ are given by 
\mbox{$Z_{f,N+1}(g) \upharpoonright \cF_n = 
\zeta_{f, n + 1}(g) \one  \upharpoonright \cF_n$}, $n \in \NN_0$. 
The operator on the right hand side of the above 
equality is a gauge invariant element of $\obfF$ and
hence contained in $\obfA$. The operator on the 
left hand side commutes with spatial translations
and $\obfA$ does not contain such operators, apart
from mulitples of the identity (this is a consequence of the
locality properties of the observables and the 
spatially asymptotic vacuum structure of the states in $\cF$). 
It follows that 
$$ 
Z_{f,N+1}(g_1 g_2) Z_{f,N+1}(g_1)^* Z_{f,N+1}(g_2)^* = \chi_f(g_1, g_2) \, \one \, ,
$$
where $\chi_f(g_1, g_2) \in \TT$. We first apply this equality to 
the vacuum vector $\bOmega$, giving     
$\chi_f(g_1, g_2) = \zeta_{f,1}(g_1 g_2) \, \overline{\zeta_{f,1}}(g_1) 
\, \overline{\zeta_{f,1}}(g_2)$. Then we apply the equality   
to all other subspace $\cF_n$, showing that 
$g \mapsto \overline{\zeta_{f,1}}(g) \zeta_{f, n+1}(g)$
are one-dimensional representations of $G$, $n \in \NN$. 
Hence, putting $\zeta_{f, 0}(g) \doteq \zeta_{f,1}(g)$, 
we conclude that $g \mapsto \overline{\zeta_{f,1}}(g) \, Z_{f,N}(g)$
defines a unitary representation of $G$ on $\cF$. Since 
$Y_f(g) \big( \overline{\zeta_{f,1}}(g) \, Z_{f,N}(g) \big)^* 
= \zeta_{f,1}(g) \, X_f(g) \in \obfA$, this proves statement (ii).

\medskip 
(ii) $\rightarrow$ (iii): 
Let $g \mapsto Z_{f,N}(g)$ be the unitary representation of $G$ with 
values in~$\bfN$, given in (ii). Making use of its eigenvalues  
$\zeta_{f,n}(g)$ on $\cF_n$,  
we define operators $U_N (g)$, putting 
$$
U_N(g) \upharpoonright \cF_n \doteq 
\overline{\zeta_{f,n}}(g) \,  \overline{\zeta_{f,n-1}}(g) \cdots 
\overline{\zeta_{f,0}}(g) 
 \upharpoonright \cF_n \, , \quad n \in \NN_0 \, .
$$
Since $g \mapsto \overline{\zeta_{f,n}}(g) \in \TT$, $n \in \NN_0$, 
are one-dimensional representations
of $G$, it follows that $g \mapsto U_N(g)$ is  
a unitary representation of $G$ on $\cF$.
Moreover, by construction 
$$
U_N(g) W_f^* \, U_N(g)^* =  W_f^* \, U_{N-1}(g)  U_N(g)^*   
= W_f^* \, Z_{f,N}(g)^* \, .
$$  
Applying to this equality the adjoint action of
$U(g)$ and multiplying it then from the left by $W_f$ leads to  
\begin{align*}
W_f V(g) W_f^* V(g)^{-1} & = W_f U(g) U_{f,N}(g) W_f^* U_N(g)^* U(g)^* \\
& = W_f U(g)  W_f^* U(g)^* \, Z_{f,N}(g)^* \doteq 
X_f(g) \in \obfA \, .
\end{align*}
Hence $\ad V(g)(W_f^*) = W_f^* X_f(g) \in \obfF$,  
and this inclusion holds also for the adjoint operators.
Since $\obfF$ is generated by $\obfA$
and $W_f, W_f^*$ and $\obfA$ is stable under the adjoint action of  
$g \mapsto V(g)$, $g \in G$, statement (iii) follows.

\medskip 
\ (iii) $\rightarrow$ (i): \ By assumption,   
$X_f(g) \doteq W_f \, \ad V(g)(W_f^*)$ are elements 
of the field algebra $\obfF$. So,
being gauge invariant, they are contained in $\obfA$, $g \in G$. 
Since the adjoint actions 
of $V(g)$ and $U(g)$ coincide on the observable algebra,
it is also clear that the partial isometries 
$X_f(g) \in \obfA$ intertwine 
$\rho_f$ and ${}^g \! \rho_f$. Moreover, they have the correct
initial and final projections, $g \in G$. 
Hence ${}^g \rho_f \simeq \rho_f$. This completes
the proof of the lemma. 
\qed \end{proof} 

Thus the answer to the question of  
whether the adjoint action of the representation $g \mapsto U(g)$, 
leaving the observable algebra $\obfA$ invariant, can be extended to the
field algebra $\obfF$ is encoded in the operators
$Y_f(g) \doteq W_f \ad U(g)(W_f^*)$, $g \in G$. 
It is affirmative if and only if 
these operators are contained in the observable algebra $\obfA$, 
possibly multiplied by some unitary representation
of $G$ which is contained in the 
von Neumann algebra $\bfN$, generated by $N$. This suggests to treat 
concrete problems according to the following scheme: 
\begin{enumerate}
\item[(1)] Check whether the restrictions of the 
operators $Y_f(g)$ to the $n$-particle subspaces, 
$Y_{f,n}(g) \doteq Y_f(g) \upharpoonright \cF_n$, satisfy 
the condition 
$Y_{f,n}(g) \in \fK_n$, $n \in \NN_0$, cf.\ relation 
\eqref{e3.1} and the remarks thereafter. 

\vspace*{0.5mm}
\item[(2)] If so, check whether the operators $Y_{f,n}(g)$ satisfy
the (generalized) 
coherence condition $\kappa_n\big(Y_{f,n}(g)\big) = \xi_n(g) Y_{f,n-1}(g)$,
where $g \mapsto \xi_n(g) \in \TT$ 
are representations (characters) of $G$, $n \in \NN_0$, cf.\ 
relation~\eqref{e3.3} and the preceding lemma.   
\end{enumerate}
If both conditions are satisfied, the isomorphism
given in relation \eqref{e3.5} and the preceding lemma
imply that there exists an extension of the 
adjoint action of $g \mapsto U(g)$ on $\obfA$
to the field algebra $\obfF$, $g \in G$.

\section{Symmetries and dynamics}
\setcounter{equation}{0}

Having explained the general framwork, we will establish now the
covariance of morphisms for a large family of symmetry
transformations and dynamics. The invariance of the field algebra
under these transformations then follows. Our main results are
presented in this section and we will also give proofs here 
in the non-interacting case. Since the analysis is 
more laborious in the
presence of interaction caused by pair potentials, we will outline 
here only the steps involved in the proof and present the technical 
details in the appendix.

\medskip
Turning to the non-interacting case, 
let $g \mapsto U_1(g)$ be a unitary representation of some
group $G$ on $\cF_1$. Any such representation can be promoted 
to a unitary representations
$g \mapsto U_\otimes(g)$ on~$\cF$ by forming $n$-fold tensor
products of $U_1(g)$ on the subspaces $\cF_n$, $n \in \NN$;
on $\cF_0$ one chooses the trivial representation.
Examples of physical interest are the spatial translations,
rotations, and the time translations induced by arbitrary non-interacting 
Hamiltonians inclusive of external potentials.   

\medskip
The operators $U_\otimes(g)$ commute with the particle number operator~$N$ 
and do not mix tensor factors in Fock space $\cF$. 
The algebra of observables is therefore stable under the adjoint action
of the unitaries, $\ad U_\otimes(g)(\obfA) = \obfA$, $g \in G$. 
This fact follows from arguments given in \cite{Bu1} which we briefly recall 
here.  
Let  $U_{\otimes,n}(g) \doteq U_\otimes(g) \upharpoonright \cF_n$ and let
$C_1, C_2, \dots, C_m \in \fC_1$ be compact operators on $\cF_1$. Then 
\begin{align*}
& U_{\otimes,n}(g) \,( C_1 \otimes_s \cdots \otimes_s C_m  
\otimes_s \underbrace{1 \otimes_s \cdots \otimes_s 1}_{n-m} ) \, U_{\otimes,n}(g)^* \\
& = 
(U_{1}(g) C_1 U_{1}(g)^*) \otimes_s \cdots \otimes_s 
(U_{1}(g) C_m U_{1}(g)^*) 
\otimes_s \underbrace{1 \otimes_s \cdots \otimes_s 1}_{n-m} \, .
\end{align*}
Since compact operators are mapped into compact operators by the 
adjoint action of any unitary operator, relations
\eqref{e3.1} to \eqref{e3.3} imply that the adjoint action of 
$U_\otimes(g)$ maps the inverse limit $\bfK$ of the inverse system 
$\{\fK_n, \kappa_n \}_{n \in \NN_0}$ into itself. The isomorphism 
given in \eqref{e3.5} then implies that $\obfA$ is stable
under the adjoint action of $U_\otimes(g)$, $g \in G$.
It also follows from the preceding equality that if $G$ is a 
topological group and $g \mapsto U_1(g)$ is continuous in the 
strong operator topology on $\cF_1$, then 
$g \mapsto \ad U_{\otimes(g),n}(K_n)$ is norm continuous
for any $K_n \in \fK_n$, $n \in \NN_0$. Thus the restrictions 
of the functions $g \mapsto \ad U_{\otimes}(g)(A)$
to the subspaces 
${\textstyle \bigoplus_{k = 0}^n} \cF_k$ are norm continuous for 
any $A \in \obfA$, $n \in \NN_0$.

\medskip 
In order to see that the action of $G$ 
by the unitaries $U_\otimes$ can be extended to the field algebra, 
we make use of the following lemma which will also be used
in the appendix.
\begin{lemma} \label{l5.1}
  Let $f_1, f_2 \in \cD(\RR^s)$ be normalized, let
  $n \in \NN_0$, and let $P_n$ be the projection onto the
  subspace $\bigoplus_{k = 0}^n f_k \in \cF$. There is a
  constant $c_n$ which does not depend on the choice of $f_1, f_2$
  such that
  $$ \| (W_{f_1}^* - W_{f_2}^*) P_n \| \leq c_n \| f_1 - f_2 \|_2 \, . $$
  Thus any sequence
  $\{ f_k \in \cD(\RR^s) \}_{k \in \NN}$ of normalized functions,  which converges
  strongly in $L^2(\RR^s)$ to  $f_{\infty}$,   
  determines a Cauchy sequence  
  $k \mapsto W_{f_k}^* P_n$ in the norm topology,
  $n \in \NN_0$. The limit of $k \mapsto W_{f_k}^*$ exists on
  $\cF$ in the strong operator topology and determines an isometric
  tensor $W_{f_{\infty}}^*$ in the field algebra $\obfF$. 
\end{lemma}
\begin{proof}
One has 
$\| \big( a(f_1) - a(f_2) \big) P_n \|
= \| a(f_1 - f_2) P_n \| \leq n^{1/2} \, \|f_1 - f_2 \|_2 $ \,
and 
$$ \| \big( a^*(f_1) a(f_1) - a^*(f_2) a(f_2) \big) P_n  \|
\leq 8n \, \|f_1 - f_2 \|_2 \, .
$$
Since $a(f) P_n = P_{n-1} a(f)$ and $a^*(f) a(f)$ 
commutes with $P_n$, it follows by a
routine computation, taking into account that
$z \mapsto (1 + z)^{-1/2}$ is analytic for
$\text{Re} z > -1$, that there is some constant
$c_n'$ such that 
\begin{align*}
  & \|(W_{f_1}^* - W_{f_2}^*) P_n \| \\
& \leq 
2 \, \|\big(a(f_1) - a(f_2)\big) P_n\| +
  c_n' \,
  \| \big(a^*(f_1)a(f_1) - a^*(f_2)a(f_2)\big) P_{n-1} \| \\
& \leq (2 n^{1/2} + 8 n c_n') \, \| f_1 - f_2 \|_2 \, , 
\end{align*}
as stated. It is then clear that $k \mapsto W_{f_k}^* P_n$, $n \in \NN_0$,
are Cauchy sequences in the norm topology  
for any given strongly convergent sequence
$\{ f_k \in \cD(\RR^s) \}_{k \in \NN}$  in $L^2(\RR^s)$. 
This implies that the sequence $k \mapsto W_{f_k}^*$,
being bounded, converges on $\cF$ in the strong operator
topology to an isometric tensor $W_{f_\infty}^*$. 
In order to see that this
limit is an element of $\obfF$, we proceed to
$W_{f_k}^* = W_f^* \, W_f W_{f_k}^*$, where $f \in \cD(\RR^s)$
is normalized. The limit of the sequence of observables 
$k \mapsto W_f W_{f_k}^*$ 
defines an element of $\obfA$. For, making use of  
the isomorphism~\eqref{e3.5}, the 
sequences $k \mapsto W_f W_{f_k}^* \upharpoonright \cF_l$
determine sequences $k \mapsto K_{k,l} \in \fK_l$
which comply for fixed $k$ with the coherence condition
$\kappa_l(K_{k,l}) = K_{k, l-1}$, $l \in \NN_0$. 
As was shown in the preceding step, the limit
$\lim_{k \rightarrow \infty} K_{k,l}$ exists in the norm topology 
and is therefore   
an element $K_{\infty, l}$ of the norm-closed space $\fK_l$. Moroever,
since the inverse maps~$\kappa_l$ (being homomorphisms) are norm continuous,
the sequence $\{ K_{\infty, l} \}_{l \in \NN_0}$
complies with the coherence condition 
$\kappa_l(K_{\infty, l})  = K_{\infty, l-1}$, $l \in \NN_0$.
Applying again the isomorphism \eqref{e3.5}, now in the inverse 
direction,
shows that  $W_f  W_{f_{\infty}}^* \in \obfA$, whence
$W_{f_\infty}^* = W_f^* \, W_f  W_{f_{\infty}}^* \in \obfF$, completing the proof. 
\qed \end{proof}

Since the normalized elements of $\cD(\RR^s)$ are strongly dense 
in the unit ball of~$L^2(\RR^s)$, it follows from this lemma that 
the partial isometries $W_{f_1} W_{f_2}^*$ are contained in 
the algebra $\obfA$ for all normalized elements 
$f_1, f_2 \in L^2(\RR^s)$. Moreover, because of the tensor
product structure of the operators $U_\otimes(g)$ one 
has \ $\ad U_\otimes(g)(W_f^*) = W_{U_1(g) f}^*$
in an obvious notation. Hence, $g \in G$,  
$$
W_f \, \ad U_\otimes(g)(W_f^*) = W_f \, W_{U_1(g) f}^* \in \obfA \, .
$$
If $G$ is a topological group and $g \mapsto U_1(g)$ is continuous
on $\cF_1$ in the strong operator topology, it likewise follows from  
the preceding lemma that the functions
$g \mapsto W_f \, \ad U_\otimes(g)(W_f^*)$ as well as their adjoints  
are norm continuous on~$\ \bigoplus_{k = 0}^n \cF_k$, 
$n \in \NN_0$. 

\medskip
These observations imply according to 
Lemma~\ref{l4.1}(ii) that the adjoint action of the unitary 
group $U_\otimes$ on the observables can be extended to the 
field algebra $\obfF$;
there is no need to correct it by some unitary group in $\bfN$. In order 
to describe also the continuity properties of this action, we introduce 
an increasing family of seminorms on $\obfF$.

\medskip
\noindent \textbf{Definition:} \ 
Let $F \in \obfF$. For any $n \in \NN_0$, put 
$$                                              
\| F \|_n \doteq \sup \, 
\{ \| F \bPhi_n \|/\| \bPhi_n \| + 
\| F^* \bPsi_n \|/\| \bPsi_n \| : \bPhi_n,  \bPsi_n\in  
\textstyle{\bigoplus_{k = 0}^n} \cF_k\} \, .
$$
This family of seminorms defines a locally convex topology on 
$\obfF$ which is weaker than the norm topology. A function 
with values in $\obfF$ is said to be lct-continuous
if it is continuous with regard to this topology. Similarly,
a subset of $\obfF$ is said to be lct-dense if it is dense
in $\obfF$ with regard to this topology. 

\medskip
It is apparent that the seminorms are symmetric, $\|F^*\|_n = \|F\|_n$.
Moreover, if $F_m, F^{\, \prime}_{m^\prime} \in \obfF$ 
are tensors, $m,m^\prime  \in \ZZ$, one has (bearing in mind the changes of 
particle numbers induced by tensors)  
$$
\| F_m F^{\, \prime}_{m^\prime} \|_n \leq 
\| F_m  \|_{n + m^\prime} \, \| F^\prime_{m^\prime}  \|_n
+ \| F_m  \|_{n} \, \| F^\prime_{m^\prime}  \|_{n - m}
\, , \quad n \in \NN_0 \, ,
$$
where $\| \, \cdot \, \|_k \doteq 0$ if $k < 0 $.
So, in particular, the product of tensor-valued 
lct-continuous functions with values in  
$\obfF$ is again lct-continuous. Making use of these 
notions, the following theorem obtains. 

\begin{theorem} \label{t5.2}
Let $G$ be a group, let $U_1$ be a unitary representation 
of $G$ on $\cF_1$, and let $U_\otimes$ be its promotion
to $\cF$, defined above. Then 
$\ad U_\otimes(g)(\obfF) = \obfF$, \mbox{$g \in G$}. 
If $G$ is a topological group and 
$g \mapsto U_1(g)$ is continuous on $\cF_1$ in the strong 
operator topology, then $g \mapsto \ad U_\otimes(g)(F)$
is lct-continuous, $F \in \obfF$.
Furthermore, if $G$ is locally compact 
there exists a lct-dense sub-C*-algebra \
$\obfF_\otimes \subset \obfF$ on which the 
adjoint action of the unitary group 
$U_\otimes$ is pointwise norm continuous. 
\end{theorem}
\begin{proof}
The first part of this statement follows from the preceding 
arguments. For the proof of 
the lct-continuity of the elements of 
$\obfF$ under the adjoint action of $U_\otimes$ 
we recall that  
any element $F \in \obfF$ and its adjoint 
$F^* \in \obfF$  can be approximated in norm by finite 
sums of tensors $F_m \in \obfF$, $m \in \ZZ$. 
Since the tensor character does not change 
under the adjoint action of~$U_\otimes$, it is 
sufficient to establish the lct-continuity for these tensors.
If $m = 0$, hence $F_0 \in \obfA$, this property was established  
prior to Lemma~\ref{l5.1}. If $m > 0$ one proceeds to 
$F_m = (F_m W_f^{* \, m}) \, W_f^m $. The operator 
$F_m W_f^{* \, m}$ is again an element of~$\obfA$, hence  
lct-continuous under the adjoint action of $U_\otimes$. Since 
$
\ad U_\otimes(g)(W_f) = (\ad U_\otimes(g)(W_f) \,  W_f^*)
\, W_f 
$,
it follows from the remarks after Lemma~\ref{l5.1} 
that $W_f$ and its adjoint are also lct-continuous 
under this action. But the 
product of tensor-valued lct-continuous functions is 
lct-continuous, which establishes the lct-continuity of~$F_m$
under the action of $\ad U_\otimes(g)$. If $m < 0$ 
one proceeds from $F_m = W_f^{* m} (W_f^m F_m)$ and a similar
argument as the preceding one establishes its lct-continuity as well. 
The lct-continuity of any $F \in \obfF$ under the adjoint action 
of $U_\otimes$ then follows.

\medskip 
If $G$ is locally compact, there exists a 
left invariant Haar measure $\mu$ on $G$, so 
one can smooth out the tensors $F_m \in \obfF$, $m \in \ZZ$.
To this end we pick any continuous
function $k : G \rightarrow \CC$ with compact support   
and proceed to the integrals (defined in the strong operator topology)
$$
F_m(k) \doteq \int \! d \mu(g) \, k(g) \, 
\ad U_\otimes(g)(F_m)  \, .
$$
The resulting functions 
$g \mapsto \ad U_\otimes(g)\big(F_m(k)\big)$
are norm continuous due to this smoothing procedure and have values in $\obfF$. 
For the proof of the latter assertion, let 
$m > 0$. We then proceed to the equality  
$$ 
F_m(k) 
= \Big( \int \! d \mu(g) \, k(g) \, \ad U_\otimes(g)(F_m) \, W_f^{* \, m} 
\Big) \, W_f^m \, .
$$ 
According to the preceding arguments,
the function $g \mapsto \ad U_\otimes(g)(F_m) \, W_f^{* \, m}$
has values in $\obfA$ and is lct-continuous. Its restrictions to 
the subspaces $\cF_l \subset \cF$ define coherent families of 
norm continuous functions in $\fK_l$, $l \in \NN_0$, so their  
integrals are defined in the norm topology. Making use again
of the facts that the spaces $\fK_l$ are norm complete and 
the inverse maps $\kappa_l$ are continuous, $l \in \NN_0$, 
the isomorphism in relation \eqref{e3.5} implies  
$\int \! d \mu(g) \, k(g) \, \ad U_\otimes(g)(F_m) \, W_f^{* \, m}
 \in \obfA$. Hence $F_m(k) \in \obfF$. 
A similar argument establishes this inclusion for $m \leq 0$.

\medskip 
Consider now the C*-algebra $\obfF_\otimes$  generated by
the operators $F_m(k)$ for arbitrary continuous functions
$k: G \rightarrow \CC$ with compact support and 
arbitrary $m \in \ZZ$. As we have shown, the unitary
group $U_\otimes$ acts norm continuously on these generating 
elements and hence pointwise on $\obfF_\otimes$.
Moreover, one can recover on any supspace 
${\textstyle \bigoplus_{l=0}^n} \cF_l$, $n \in \NN_0$, the 
original tensors 
$F_m$ in the norm topology from the mollified operators $F_m(k)$, letting
the measures 
$d\mu(g) k(g)$ tend to the Dirac measure at the unit element of $G$. 
Since the tensors $F_m$, $m \in \ZZ$, generate $\obfF$, 
the algebra~$\obfF_\otimes$ is lct-dense in $\obfF$, completing the
proof.  \qed \end{proof}

The preceding theorem shows that symmetry transformations
and dynamics of physical interest, involving arbitrary external 
forces, act as automorphisms on the 
field algebra $\obfF$ and have strong continuity properties. 
Moreover, one can proceed to subalgebras of $\obfF$ on which
this action is norm continuous, yielding C*-dynamical systems.
In the sequel we will demonstrate that these desirable features  
of the field algebra persist in the presence of interaction. 
For the sake of concreteness and limitation of technical  
difficulties, we restrict our attention  
to gauge invariant Hamiltonians $H$ which are defined 
on their standard domains of definition in $\cF$ by 
\begin{equation} \label{e5.1}
\begin{gathered} 
H = \int \! d\bx \, \bpartial a^*(\bx) \, \bpartial a(\bx) 
+ \kappa^2 \! \int \! d\bx \, \bx^2 \ a^*(\bx) \, a(\bx) \\
+ \int \! d\bx \! \! \int \! d\by \ a^*(\bx) a^*(\by) \, V(\bx - \by) \,
a(\bx) a(\by) \, .
\end{gathered}
\end{equation}
Here $\bpartial$ denotes the gradient and $V \in C_0(\RR^s)$ is a  
real, continuous, and symmetric function vanishing at infinity 
(describing a pair potential). In order to cover also trapped systems, 
we admit an external harmonic force, scaled by~$\kappa^2 \geq 0$.
These Hamiltonians cover a large class of  
attractive and repulsive two-body potentials, including potentials of   
long range. In order to reduce the technicalities, we exclude
potentials having singularities; methods to treat such 
potentials in the framwork of the resolvent
algebra were discussed in \cite[Sec.~6]{Bu2}.
We also do not consider non-harmonic trapping forces,
which can be handled by refinements of the present 
arguments. The assumption of gauge invariance of the Hamiltonians is,
however, essential for the present approach. In the last section 
we will therefore comment on the 
treatment of non-gauge invariant Hamiltonians,  
involving for example an additional linear term in the field, 
such as in the Nelson model and similar theories, cf.\ \cite{Sp}.

\medskip
The above Hamiltonians $H$ are selfadjoint, so given 
any such $H$ we can proceed to
the unitary group $t \mapsto e^{itH}$ of time translations on $\cF$.
As was shown in~\cite{Bu1}, the observable algebra
$\obfA$ is stable under the corresponding adjoint action 
$\ad e^{itH}$, $t \in \RR$. Hence we can lift it to the morphism 
$\rho_f = \ad \, W_f$,
cf.\ \eqref{e4.1}, where 
the normalized function $f \in \cD(\RR^s)$
will be kept fixed in the subsequent analysis. The lift 
is given by, cf.\ \eqref{e4.2},    
$$
{}^{t \!} \rho_f \doteq \ad e^{itH} \, \scirc \, \rho_f \,
\scirc \, \ad e^{-itH} \, , \quad t \in \RR \, ,
$$
and there exist intertwining operators $W_f \ad e^{itH}(W_f^*)$
between the morphisms $\rho_f$ and ${}^{t \!} \rho_f$, $t \in \RR$.
Thus, in order to prove that the time translations can be
extended to the field algebra $\obfF$ we need to
show that these intertwining operators comply with the two
conditions given at the end of the preceding section.
For notational convenience we consider here their adjoints  
\begin{equation} \label{e5.2}
\ad e^{itH} (W_f) \, W_f^* = e^{itH}  W_f  e^{-itH}  W_f^*
= e^{itH}  e^{-it W_f H W_f^*} \, E_f \doteq \Gamma_f(t) \, E_f \, ,
\end{equation}
where $\Gamma_f(t)$ denotes the product of the two exponentials
on the right hand side of the second equality.
Recalling that $E_f = W_f W_f^*$ is the gauge invariant projection
onto the orthogonal complement of the kernel of $a(f)$, it follows
that $E_f \big( \bigoplus_{k = 0}^n \cF_k \big) = | f \rangle \otimes_s 
\big( \bigoplus_{k = 0}^{n-1} \cF_k \big)$, $n \in \NN_0$, where we identify
$| f \rangle \otimes_s \Omega$ with $| f \rangle \in \cF_1$; thus, in 
a somewhat sloppy notation, $E_f \cF = | f \rangle \otimes_s \cF$.  
It will be crucial that the unitaries $\Gamma_f(t)$ 
are restricted to this subspace of Fock space, where at least
one particle is localized in a fixed region, the support
of $f$. We proceed then as follows.

\medskip
In a first step we restrict the unitaries $\Gamma_f(t)$ 
to the  subspaces of $E_f \, \cF$ with fixed particle number
$n \in \NN_0$, \ 
$\Gamma_{f,n}(t) \doteq \Gamma_f(t) \, E_{f,n}$,
where $E_{f,n}$ coincides with $E_f$ on $\cF_n$
and vanishes on $\cF_m$, $m \in \NN_0 \backslash \{ n \}$.
Note that the range of the operators $\Gamma_{f,n}(t)$ is 
contained in $\cF_n$, but not in
$E_{f,n} \, \cF_n = |f\rangle \otimes_s \cF_{n-1}$. 
In order to see that these restricted operators are elements of 
$\fK_n$, we analyze the difference between the 
generators of the underlying unitary groups, 
\begin{equation} \label{e5.3}
(H - W_f H W_f^*) \upharpoonright E_f \, \cF_n  =
(H_n -  W_f H_{n-1} W_f^*)  \, E_{f,n} \, .
\end{equation}
Because of the choice of $f$, these operators are
densely define. As a matter of fact, these differences 
are bounded operators. This puts us into the position to 
expand the unitaries $\Gamma_{f,n}(t)$ into a norm-convergent
Dyson series of time ordered integrals, involving the
interacting dynamics. 

\medskip 
Parts of the operators appearing in the difference between 
the generators are elements of $\fK_n$;
it then follows from results in \cite{Bu1} that their
time ordered integrals also belong to this algebra. But there appear 
also terms which are not of this type.  We shall show in the second step
of our argument that the time ordered integration     
improves the properties of these terms so that the integrated 
terms are likewise contained in $\fK_n$. 
Since the integrals involve the interacting dynamics, this step
requires another expansion of Dyson type, where the interacting
dynamics is expanded in terms of time ordered integrals 
involving the action of the non-interacting dynamics. 
(It is at this point where the restrictions on the trapping potential
lead to simplifications.)
These results, together with the convergence of the 
Dyson expansion, show that
\mbox{$\Gamma_{f,n}(t) \in \fK_n$}, $t \in \RR$, for
$n \in \NN_0$. It completes the proof of the first point in 
the list at the end of the preceding section. 

\medskip 
In the final step we check whether the operators 
$\Gamma_{f,n}(t) \in \fK_n$ are restrictions
of some observable on $\cF$ to the subspaces  
$\cF_n$, $n \in \NN_0$. To this end we apply to these operators
the inverse maps $\kappa_n$. In order to gain control on  
their respective images we need to prove that
one may interchange the action of 
$\kappa_n$ with the time ordered integrals involved in the 
computation of~$\Gamma_{f,n}(t)$. This requires some further analysis.
It results in the desired relation 
$\kappa_n\big(\Gamma_{f,n}(t)\big) = \Gamma_{f,n-1}(t)$, $n \in \NN_0$     
and thereby establishes the second point in the above check list. 
There is no need to modify the 
operators by characters of~$\RR$ in order to arrive 
at this conclusion. 
Lemma~\ref{l4.1} then implies that the field algebra 
$\obfF$ is stable under the action of 
$\ad U(t)$, $t \in \RR$. In the course of this analysis 
we keep also control on the continuity properties of the functions   
$t \mapsto \Gamma_{f,n}(t)$, $n \in \NN$, and thereby arrive at the
subsequent theorem. The details of proof are given in the appendix.

\begin{theorem} \label{t5.3}
  Let $t \mapsto e^{itH}$ be the unitary group on $\cF$
  which is determined by a Hamiltonian
  of the form given in \eqref{e5.1}. The
  adjoint action of this group leaves the field algebra
  invariant, $\ad e^{itH}(\obfF) = \obfF$, $t \in \RR$, and 
  the resulting functions $t \mapsto \ad e^{itH}(F)$ are
  lct-continuous, $F \in \obfF$. Moreover, there is a lct-dense
  sub-C*-algebra \, $\obfF_0 \subset \obfF$
  on which this action is pointwise norm continuous. 
\end{theorem}  

\section{Conclusions}
\setcounter{equation}{0}

We have completed here our construction of a 
C*-algebraic framework for infinite non-relativistic Bosonic systems 
and their dynamics. In a preceding step~\cite{Bu1}, 
we clarified the properties of the gauge 
invariant (particle number preserving) observables, which 
are elements of the 
resolvent algebra of a non-relativistic Bose field. 
They generate a C*-algebra $\obfA$ with a surprisingly simple structure: 
it is a (bounded) projective limit of the direct sum of    
observable algebras for finite particle number. The latter algebras 
are built from compact operators on the single particle space, tensored 
with unit operators. Such structures were also found by Lewin~\cite{Le}. 
Yet some important feature is missing in that 
analysis; namely the existence of homomorphisms, relating the 
algebras for different 
values of the particle number. These homomorphisms were constructed 
in \cite{Bu1}, making use of clustering properties
of the states on the resolvent algebra in Fock space. 
They were a vital ingredient in the proof 
that a large family of dynamics, involving two-body 
potentials, acts by automorphisms on $\obfA$. 

\medskip
In the present article we have exhibited operators in the resolvent algebra,
which transform as tensors under the action of the 
gauge group. The construction is based on a weak form of harmonic
analysis, the crucial point being that the resultant tensors are 
still elements of the original resolvent-C*-algebra. Extending the
method of construction used for the observable algebra to these tensors, 
we have obtained a C*-algebra $\obfF$, the field algebra. It  
is generated by the observables in $\obfA$  and a single additional 
isometric tensor $W$.
The choice of this tensor is largely arbitrary within certain limitations;
irrespective of its choice, 
one arrives at the same algebra. So also the field algebra $\obfF$ has a  
concrete and simple structure. 

\medskip
In order to reveal the importance of the basic isometric tensors $W$,
we adopted ideas from sector analysis in relativistic quantum 
field theory \cite{DoHaRo}. 
In the present non-relativistic setting, the algebra of 
observables $\obfA$ 
gives rise to disjoint (super\-selected) representations  
on Hilbert spaces with different particle numbers. 
Akin to the relativistic setting, these representations
are related by particle number decreasing morphisms 
$\rho$ of the algebra of observables, $\rho(\obfA) \subset \obfA$. 
They are induced by the basic isometries, $\rho = \ad W$. As in the 
relativistic setting, different morphisms 
$\rho_1 = \ad W_1$, $\rho_2 = \ad W_2$ are related by intertwining 
operators $W_1 W_2^*$, which are contained in the algebra 
of observables. This equivalence 
expresses the fact that the morphisms generate 
equivalent representations on any given 
representation space of $\obfA$. They describe 
indistinguishable sets of states, where a particle 
has been removed from the states in the initial representation.

\medskip
These insights were then applied to symmetry transformations 
$g$ of the observables which, in the case at hand, were primarily of 
geometric nature. Examples are spatial translations, rotations and the time 
translations. They 
act on the algebra of observables by automorphisms, which can be lifted
to the morphisms, $\rho \rightarrow  {}^g\rho$. In other words, the morphisms 
can be translated, rotated,
and time shifted. It is then an obvious question whether a transformed 
morphism~${}^g\rho$ and the original morphism $\rho$ 
lead on any given representation space to equivalent representations, \ie 
whether the geometric operations on the morphisms do not 
alter the resulting sets of states. Again, the answer is 
affirmative if the initial morphisms and 
the transformed ones are related by intertwining operators, which are 
contained in the algebra of observables. Such morphisms are said to
be covariant under the respective symmetry 
transformations \cite{DoHaRo}. 
Any representation of $\obfA$, in which the symmetry transformations
are unitarily implemented, then gives rise to a representation 
$\rho$ of $\obfA$ on the same 
space, in which these transformations are also unitarily 
implemented. Thus the invariance of a 
representation of $\obfA$ under some symmetry transformation
is preserved by the action of 
covariant morphisms $\rho$, and the corresponding
unitary operators implementing this symmetry 
are related by intertwiners in $\obfA$.

\medskip
We have restricted here our attention to symmetries which are unitarily 
implemented on Fock space and induce automorphisms of the algebra of 
observables. Making use of the fact that the observable 
algebra $\obfA$ is faithfully represented on Fock space, we could specify  
candidates for the corresponding intertwining operators between the 
morphisms, making use of the basic isometries $W$. Even though 
these candidate intertwiners exist as bounded operators on Fock space,
their existence does not answer the preceding question. An
affirmative answer requires that the intertwiners are 
elements of the C*-algebra of observables $\obfA$, which is only a 
small subalgebra of the algebra of all bounded operators on Fock space. 

\medskip 
The detailed analysis 
of the candidate intertwiners therefore constituted a substantial part of the 
present investigation. In case of basic symmetries, such as 
spatial translations, rotations, and non-interacting
time translations, it was not difficult to show 
that the candidate intertwiners are elements of $\obfA$.
In the more interesting case of interaction, we have established
this fact for dynamics involving arbitrary continuous two body potentials, 
vanishing at infinity. We believe that these results can be extended 
with some effort to singular potentials, such as the Coulomb potential 
and the Yukawa potential. 

\medskip 
Whenever the candidate intertwiners between the morphisms belong to 
the observable algebra, the unitaries implementing the 
action of the respective symmetry
transformations on Fock space define by their adjoint action automorphisms 
of the field algebra $\obfF$. So this kinematical C*-algebra, which is 
generated by the basic canonical operators underlying the theory, 
is compatible with the Heisenberg picture for a large family
of dynamics. More precisely, for all initial data in~$\obfF$, 
the corresponding solutions of the Heisenberg equation lie also 
in~$\obfF$.

\medskip 
The Heisenberg picture is particularly 
useful in case of infinite systems. For, states of physical interest, 
such as equilibrium states at different temperatures, in general 
require different, disjoint Hilbert space representations. 
In contrast, the C*-algebra $\obfF$ is fixed. 
Furthermore, the generators implementing a given dynamics 
on a representation space (such as the Liouvillians) 
depend in general on the underlying states, wheras its action 
on $\obfF$ is defined in a state independent manner. 
This fact leads 
for example to simplifications in the treatment of stationary states,
where the generators of the dynamics can be determined from the 
algebra $\obfF$ in the
corresponding  GNS representations by standard methods. 
So the algebra $\obfF$ provides a convenient framework which is superior to 
the Weyl algebra setting. The latter algebra admits only rather trivial  
dynamics given by symplectic transformations. 

\medskip 
There is, however, a point which deserves further studies. In the present
analysis we have considered dynamics, which are particle number preserving
(gauge invariant) and hence leave the algebra of observables
invariant. Yet, thinking of models, where the Bose field is coupled to
other quantum systems, this feature may fail. One frequently models such 
a situation by Hamiltonians of the form considered in
the present article with an additional term which is linear in the basic 
field. It is an open problem whether such dynamics, which 
do not preserve the observable algebra $\obfA$, still preserve the 
field algebra $\obfF$. 

\medskip
Thinking of the Trotter product formula, 
one may study this problem by looking at the alternating  
adjoint actions of the exponentials of a regularized  
field (Weyl operator) and of a gauge invariant Hamiltonian
on the algebra $\obfF$. 
As we have shown, gauge invariant Hamiltonians induce
automorphisms of $\obfF$, so one needs to show that this
algebra is also stable under the adjoint action of Weyl operators. 
It is not difficult to see that the resolvent algebra is 
stable under such action, inducing on 
the basic resolvents the maps
$$
R(\lambda,f) \mapsto R(\lambda +i \mu(f),f) \, , \quad 
\lambda \in \RR \backslash \{ 0 \} \, , f \in \cD(\RR^s) \, ,
$$
where $\mu : \cD(\RR^s) \rightarrow \RR$ is a real linear 
functional on the test function space. It follows from the 
remarks made after Theorem~3.6
in \cite{BuGr1} that the transformed resolvents are again elements
of the resolvent algebra. In order to cover also interacting systems,
one has to establish this result for the field algebra $\obfF$, 
however, which seems feasible. 

\medskip
Note that the automorphic actions of Weyl type are also meaningful
in cases, where the functional 
$\mu$ is a distribution and a corresponding Weyl operator
no longer exists. This is for example of interest in situations, where one 
describes condensates of Bosons in infinite space. In simple 
cases one can describe their presence by a spatially homogenous 
functional $\mu$. Thinking of models with repulsive two-body
potentials, where the Fock vacuum is still a ground state, one may therefore 
hope that similar transformations   
describe the vacuum in presence of a condensate also in 
case of interaction. 

\medskip
Lastly, in models where the field is coupled with other quantum 
systems, there often appear linear terms involving the field 
operator only in the interaction operators, 
whereas the remaining parts describe the free 
evolution of the subsystems. Examples are models 
of Nelson type \cite{Sp}, which are frequently used in studies of 
infrared problems, involving massless 
particles in infinite space. There a more direct 
approach to the proof of
the stability of $\obfF$, tensored with the algebra of the
coupled quantum system, seems possible.  
In fact, expansions of Dyson type seem to lead to the desired result. 
We hope to get back to these problems of continuing 
physical interest by making use of the present novel approach.  

\appendix
\section{Appendix}

In this appendix we give the proof of Theorem~\ref{t5.3}, 
carrying out the various steps outlined prior to its statement.
We begin by introducing the notation used in what follows. 

\vspace*{-5mm}
\subsection{Fields and particle picture} \label{ssa.1} \hfill

\noindent Since we will freely alternate between the field theoretic approach 
and the particle picture, based on the interpretation of Fock space, 
let us recall some standard formulas.
Given $f_1, \dots, f_n \in \cD(\RR^s)$ one has for the 
symmetric tensor
product of the corresponding single particle vectors the relation 
$$
|f_1 \rangle \otimes_s \cdots \otimes_s  | f_n \rangle =
(1 / n!)^{1/2} \ a^*(f_1) \cdots a^*(f_n) \, \bOmega \in \cF_n \, . 
$$
Next, let $O_1$ be a single particle operator on $\cF_1$
with (distributional) kernel 
\mbox{$\bx, \by \mapsto \langle \bx | O_1 | \by \rangle$}. 
Its canonical lift to 
$\cF_n$, $n \in \NN_0$, obtained by forming symmetrized tensor products 
with the unit operator and amplifying it with the appropriate weight factor 
$n$, is given by   
$$
n \ (O_1 \otimes_s \underbrace{1 \otimes_s \cdots \otimes_s 1}_{n-1}) 
= \int \! d\bx \! \int \! d\by \, a^*(\bx) \,
\langle \bx | O_1 | \by \rangle \, a(\by) \upharpoonright \cF_n \, .
$$
The field theoretic operator on the right hand side of this equality 
will be called \textit{second quantization} of~$O_1$.
Similarly, if $O_2$ is a two-particle operator acting on 
$\cF_2$ with kernel 
$\bx_1, \bx_2, \by_1, \by_2 \mapsto 
\langle \bx_1, \bx_2 | O_2 | \by_1, \by_2 \rangle$, one has, $n \in \NN_0$, 
\begin{align*}
& n(n-1) \ 
(O_2 \otimes_s \underbrace{1 \otimes_s \cdots \otimes_s 1}_{n-2}) \\
& =  
\int \! d\bx_1 \!  \! \int \! d\bx_2 \!  \! 
\int \! d\by_1 \!  \! \int \! d\by_2 \, 
a^*(\bx_1) a^*(\bx_2) \,
\langle \bx_1 , \bx_2 | O_2 | \by_1, \by_2 \rangle \, a(\by_1) a(\by_2) 
\upharpoonright \cF_n \, . 
\end{align*}
The operator on the right hand side 
will be called \textit{second quantization} of $O_2$.

\medskip
The Hamiltonians of interest here, given in
equation \eqref{e5.1}, have the form
\begin{align*}
 H & = \int \! d\bx \, 
\big( \bpartial a^*(\bx) \, \bpartial a(\bx) 
+ \kappa^2 \bx^2 \, a^*(\bx) \, a(\bx) \big)  \notag \\ 
& + \int \! d\bx \! \! \int \! d\by \ a^*(\bx) a^*(\by) \, V(\bx - \by) \,
a(\bx) a(\by) \, .
\end{align*}
The first integral is the second quantization 
of the operator $\bP_\kappa^2 \doteq \bP^2 + \kappa^2 \, \bQ^2$
on $\cF_1$, where $\bP$ is the momentum and $\bQ$ the 
position operator;
the second integral is the second quantization of the 
two-particle potential $V$ on~$\cF_2$. 
Note that the kernel of proper pair potentials 
on $\cF_2$ has the singular form 
$$
\bx_1, \bx_2, \by_1, \by_2 \mapsto 
(1/2) \, (\delta(\bx_1 - \by_1) \delta(\bx_2 - \by_2) + 
\delta(\bx_1 - \by_2)  \delta(\bx_2 - \by_1) ) \, V(\by_1 - \by_2) \, ,
$$
which reduces the second quantization of $V$ to a double integral.
(We will have occasion to consider also less singular potentials
whose second quantization requires more integrations.) 
Given $n \in \NN_0$, the restriction 
$H_n \doteq H \upharpoonright \cF_n$ can 
thus be presented in the form 
\begin{equation} \label{ea.1} 
H_n =  
\  n \, (\bP_\kappa^2 \otimes_s \underbrace{1 \otimes_s
  \cdots \otimes_s 1}_{n-1}) + \ n(n-1) \, 
(V \otimes_s \underbrace{1 \otimes_s \cdots \otimes_s 1}_{n-2})  \, .
\end{equation}
This version will be useful in our subsequent analysis,
where we need to decompose the operators $\bP_\kappa^2$ 
and $V$ into different pieces in order to relate them to 
elements of the algebras $\fK_n$, cf.~Eqn.~\eqref{e3.1}

\medskip
We will also make use of the second quantization $N_f$ of the 
one-particle operator $E_{f,1}$, the projection onto the ray of
$|f \rangle$ in $\cF_1$.
The restriction of this number operator to $\cF_n$ is given by
$ N_{f,n} \doteq N_f \upharpoonright \cF_n = n \,
(E_{f,1} \otimes_s \underbrace{1 \otimes_s \cdots \otimes_s 1}_{n-1})$. 
Hence all bounded functions of $N_{f,n}$ are elements of 
$\fK_n$, cf.\ Eqn.~\eqref{e3.1}.
We also note that the projection
$E_{f,n} = E_f \upharpoonright \cF_n$ is a function of this kind and can 
be expressed in terms of $E_{f,1}$ by the formula
$$
E_{f,n} = \one_n - \underbrace{(1-E_{f,1}) \otimes_s \cdots \otimes_s
(1 - E_{f,1}) }_n \, ,
$$
where $\one_n$ is the unit operator on $\cF_n$.

\vspace*{-4mm}
\subsection{Comparison of Hamiltonians} \label{ssa.2} \hfill

\noindent 
As outlined in the main text, we need to consider   
the difference of Hamiltonians  
$(H_n - W_f H_{n-1} W_f^*) \, E_{f,n}$,
cf.\ Eqn.~\eqref{e5.3}.
In our first technical lemma we focus on the second 
term in this difference and compute lifts of operators on 
$\cF_{n-1}$ to $E_{f,n} \, \cF_n \subset \cF_n$, which are 
induced by the adjoint action of $W_f$; recall that 
$E_{f,n} \, \cF_n = |f\rangle \otimes_s \cF_{n-1}$, 
$n \in \NN$. In the statement of the lemma there appear  
similarity transformations $\sigma_f$ of gauge invariant operators 
$O$ on $\cF$, given by
\begin{equation} \label{ea.2}
\sigma_f(O) \doteq (1 + N_f)^{-1/2} \, O \, (1 + N_f)^{1/2} \, .
\end{equation}
We put $\sigma_{f,n}$ for the restriction of $\sigma_f$ to 
gauge invariant operators on $\cF_n$, $n \in \NN_0$.

\begin{lemma}  \label{la.1} 
Let $n \in \NN$ and let $O_{n-1}$ be an operator with domain 
of definition $\cD_{n-1} \subset \cF_{n-1}$ which is stable under
the action of the spectral projections of $N_{f,n-1}$.
Then $\sigma_{f, n-1}(O_{n-1})$ and $\sigma_{f, n-1}^{-1}(O_{n-1})$
are defined on $\cD_{n-1}$. Moreover, one has for
any $\bPhi_{n-1} \in \cD_{n-1}$ the equalities 

\medskip
(i) \quad $W_f O_{n-1} W_f^* \ 
\big( | f \rangle \otimes_s \bPhi_{n-1} \big) \ = \ 
| f \rangle \otimes_s \sigma_{f,n-1}(O_{n-1}) \, \bPhi_{n-1}$ 

\medskip
\hspace*{-2mm} (ii) \quad $| f \rangle \otimes_s O_{n-1} \bPhi_{n-1} 
 \ = \ W_f \, \sigma_{f,n-1}^{-1}(O_{n-1}) \, W_f^* \  
\big( | f \rangle \otimes_s \bPhi_{n-1} \big)$. 
\end{lemma}
\begin{proof}
Noticing that the spectral decompositions
of $(\one_{n-1} + N_{f,n-1})^{\pm 1/2}$ are finite 
linear combinations of the spectral projections of $N_{f,n-1}$, the statement
concerning the domains of the similarity transformed operators follow.
For the proof of 
(i) we note that
$a(f) a^*(f) \, \bPhi_{n-1} = (\one_{n-1} + N_{f,n-1}) \, \bPhi_{n-1}$,  
hence 
$$
a(f) \ \big( | f \rangle \otimes_s \bPhi_{n-1} \big) =
a(f) \ n^{-1/2} a^*(f) \, \bPhi_{n-1} =
n^{-1/2} \, (\one_{n-1} + N_{f,n-1}) \, \bPhi_{n-1} \, .
$$
Thus, by the spectral properties of $N_{f,n-1}$, 
the vector $W_f^* \ \big( | f \rangle \otimes_s \bPhi_{n-1} \big) $
is also an element of~$\cD_{n-1}$. So one has  
\begin{align*}
W_f \, O_{n-1} \, & W_f^* \ \big( | f \rangle  \otimes_s \bPhi_{n-1} \big) \,
 \, = \, n^{-1/2} \, W_f \, O_{n-1} \, (\one_{n-1} + N_{f,n-1})^{1/2} \,
\bPhi_{n-1} \\
& = n^{-1/2} \, a^*(f) \, \sigma_{f,n-1}(O_{n-1}) \, \bPhi_{n-1} 
=  | f \rangle \otimes_s  
  \sigma_{f,n-1}(O_{n-1}) \, \bPhi_{n-1} \, ,
\end{align*}
proving the first statement. Statement 
(ii) follows from (i) if one replaces the operator $O_{n-1}$ by 
$\sigma_{f,n-1}^{-1}(O_{n-1})$, completing the proof.
\qed \end{proof}

We consider now the Hamiltonians $H_{n-1}$, $n \in \NN$. 
For them the spaces  
$$
\cD_{n-1} \doteq \underbrace{\cD(\RR^s) \otimes_s \cdots 
 \otimes_s \cD(\RR^s)}_{n-1} \subset \cF_{n-1}
$$ 
are domains of essential selfadjointness. In view of the
choice of the function $f$, it is also evident that these spaces
are stable under 
the action of the spectral projections of $N_{f,n-1}$. So 
the first part of the preceding lemma applies to 
$W_f \, H_{n-1} \, W_f^*$, giving the equality 
$$
W_f \, H_{n-1} \, W_f^* \upharpoonright
| f \rangle \otimes_s \cD_{n-1} 
= | f \rangle \otimes_s \big(\sigma_{f,n-1}(H_{n-1}) \upharpoonright \cD_{n-1}
\big)  \, .
$$
We compare now the operators $H_{n-1}$ and $\sigma_{f,n-1}(H_{n-1})$.
\begin{lemma} \label{la.2}
Let $n \in \NN$. Then 
$$
H_{n-1} - \sigma_{f,n-1}(H_{n-1}) 
= \check{A}_{f, n-1} + \check{B}_{f, n-1} \, .
$$
Here 
$\check{A}_{f. n-1} = \check{A}_f \upharpoonright \cF_{n-1}$, where  
$\check{A}_f = \big(\check{O}_f - \sigma_f(\check{O}_f) \big)$ 
and $\check{O}_f$ is the second  
quantization of one- and two-particle operators of finite rank;
so \mbox{$\check{A}_{f, n-1} \in \fK_{n-1}$}. If $n \geq 3$ one has 
$\check{B}_{f,n-1} = \check{B}_f \upharpoonright \cF_{n-1}$, 
where $\check{B}_f =  \big( \check{V}_f - \sigma_f(\check{V}_f) \big)$ 
and $\check{V}_f$ is the second quantization 
of the localized pair potential $V$ on $\cF_2$.
This localized potential is given by 
$$
\check{V}_{f,2} = 2 \, (E_{f,1} \otimes_s 1)  \, V  \, 
(E_{f,1}^\perp \otimes_s 1) + 2 \,
(E_{f,1}^\perp \otimes_s 1) \,   V  \, 
(E_{f,1} \otimes_s 1) \, ,
$$
where $E_{f,1}^\perp \doteq (1 - E_{f,1})$. So 
the restriction of the corresponding second quantized operator  
$\check{V}_{f, n-1} = \check{V}_f \upharpoonright \cF_{n-1}$ is 
$$
\check{V}_{f, n-1} 
= (n-1)(n-2) \, 
(\check{V}_{f,2} \otimes_s \underbrace{1 \otimes_s \cdots \otimes_s 1}_{n-3})
\, , 
$$
and the resulting operator $\check{B}_{f, n-1}$ is bounded. 
\end{lemma}

\noindent \textbf{Remark:} Since the operator $\check{V}_{f,2}$
is not an element of $\cK_2$, it has to be treated separately. 
It will be crucial in the subsequent analysis that $\check{V}_{f,2}$
is effectively localized by the factor  $(E_{f,1} \otimes_s 1)$ next to $V$.
 
\begin{proof} According to relation \eqref{ea.1} we have 
$$
H_n = n \, (\bP_\kappa^2 \otimes_s \underbrace{1 \otimes_s
  \cdots \otimes_s 1}_{n-1})
+ n(n-1) \, (V \otimes_s \underbrace{1 \otimes_s \cdots \otimes_s 1}_{n-2}) \, .
$$ 
We decompose the operator $\bP_\kappa^2$, acting in $\cF_1$, into 
$$
\bP_\kappa^2 = E_{f,1}^\perp \, \bP_\kappa^2 \, E_{f,1}^\perp
+ E_{f,1}  \, \bP_\kappa^2 \, E_{f,1}^\perp 
+ E_{f,1}^\perp  \, \bP_\kappa^2 \, E_{f, 1}
+ E_{f,1}  \, \bP_\kappa^2  \, E_{f, 1} \, .
$$
This decomposition is meaningful since $| f \rangle$ lies in the 
domain of $\bP_\kappa^2$.  The first operator on the right hand side of 
this equality maps the orthogonal complement of the ray of 
$| f \rangle$ into itself; the three remaining operators are of 
rank one. Similarly, we decompose the pair potential $V$ 
on $\cF_2$ into 
\begin{align*}
V & = (E_{f,1}^\perp \otimes_s E_{f,1}^\perp) \, V \, 
(E_{f,1}^\perp \otimes_s E_{f,1}^\perp) 
- (E_{f,1} \otimes_s E_{f,1}) \, V \, (E_{f,1} \otimes_s E_{f,1}) \\
& - \, (E_{f,1} \otimes_s E_{f,1}) \, V \, ((1 - 2 E_{f,1}) \otimes_s 1)
- ((1 - 2 E_{f,1}) \otimes_s 1) \, V \, (E_{f,1} \otimes_s E_{f,1}) \\
& + \, 2 \, (E_{f,1} \otimes_s 1) \, V \, (E_{f,1}^\perp \otimes_s 1)
+ 2 \, (E_{f,1}^\perp \otimes_s 1) \, V \, (E_{f,1} \otimes_s 1) \, .
\end{align*}
The first operator on the right hand side of
this equality maps the orthogonal complement of 
$| f\rangle \otimes_s \cF_1 \subset \cF_2$ into itself. The 
second up to the fourth terms are operators of finite rank
due to the appearance of the factor ($E_{f,1} \otimes E_{f,1}$). 
The two terms in the last line form the 
localized pair potential $\check{V}_{f,2}$   
in the statement of the lemma. 
 
\medskip 
Tensoring these operators with unit operators $1$ and multiplying them 
with appropriate factors of $n$, we obtain a corresponding decomposition
of the operator
$\vartheta_{n-1} \doteq H_{n-1} - \sigma_{f,n-1}(H_{n-1})$.
Since the operators  
\begin{eqnarray*} 
& E_{f,1}^\perp \, \bP^2 \, E_{f,1}^\perp \otimes_s
\underbrace{1 \otimes_s  \cdots \otimes_s 1}_{n-2} \, , \\  
& (E_{f,1}^\perp \otimes_s E_{f,1}^\perp) \,  V \, (E_{f,1}^\perp \otimes_s 
E_{f,1}^\perp) \otimes_s \underbrace{1 \otimes_s \cdots \otimes_s 1}_{n-3} 
\end{eqnarray*}
commute with $N_{f,n-1}$ and consequently 
stay fixed under the action of the similarity transformation $\sigma_{f,n-1}$,
they do not contribute to $\vartheta_{n-1}$ and can be omitted
from $H_{n-1}$. The remaining terms in $H_{n-1}$ consist of two types.
The first one is, for any $n \in \NN$,  
a sum of fixed one- and two-particle operators of finite rank 
which are tensored with unit operators and multiplied by  
appropriate factors of~$n$. Denoting by $\check{O}_f$ the second
quantization of these one- and two-particle operators,
it follows from Eqn.~\eqref{e3.5} that
$\check{O}_{f, n-1} = \check{O}_f \upharpoonright \cF_{n-1} \in \fK_{n-1}$.
Since $(\one_{n-1} + N_{f,n-1})^{\pm 1/2} \in \fK_{n-1}$ it is
also clear that $\sigma_{f,n-1}(\check{O}_{f, n-1}) \in \fK_{n-1}$.

\medskip
The second type of terms in~$H_{n-1}$ which contribute to 
$\vartheta_{n-1}$ arise from 
the second quantization $\check{V}_f$ of the 
localized pair potential $\check{V}_{f,2}$.
The resulting operators 
$\check{V}_{f,n-1} = \check{V}_f \upharpoonright \cF_{n-1}$ 
and their similarity transformed images $\sigma_{f,n-1}(\check{V}_{f,n-1})$
are bounded since the pair potential $V$ and the  
operators $(\one_{n-1} + N_{f,n-1})^{\pm 1/2}$ are bounded. \qed
\end{proof}

Next, we compare the operators \
$H_n \upharpoonright  |f\rangle \otimes_s \cD_{n-1} $
and \ $ |f \rangle \otimes_s \big( H_{n-1} \upharpoonright \cD_{n-1} \big) $.

\begin{lemma} \label{la.3} 
Let $n \in \NN$. One has in $\cF_n$
the equality (pointwise on $\cD_{n-1}$)
$$ 
H_n \upharpoonright |f\rangle \otimes_s \cD_{n-1}
- | f \rangle \otimes_s \big( H_{n-1} \upharpoonright  \cD_{n-1} \big)
= \hat{A}_{f,n} + \hat{B}_{f,n} \, .
$$ 
Here $\hat{A}_{f,n} = 
\hat{A}_f \upharpoonright \cF_n \in \fK_n$, where    
$\hat{A}_f = \hat{O}_f N_f^{-1} E_f$ and $\hat{O}_f$ is
the second quanti\-zation of a 
one-particle operator of rank one. If~$n \geq 2$ one has 
\mbox{$\hat{B}_{f,n} = \hat{B}_f \upharpoonright \cF_n$}, where 
$\hat{B}_f = \hat{V}_f  N_f^{-1} E_f$ and $\hat{V}_f$ is the 
second quantization of the localized pair potential 
$\, \hat{V}_{f,2} = V \, (E_{f,1} \otimes_s 1)$. 
Its restriction 
$\ \hat{V}_{f,n} = \hat{V}_{f} \upharpoonright \cF_n$ 
is given by 
$\hat{V}_{f,n} = n(n-1) \, \big( \hat{V}_{f,2} \otimes_s 
\underbrace{1 \otimes_s \cdots \otimes_s 1}_{n-2} \big) $, 
so the operator 
$\hat{B}_{f,n} = \hat{V}_{f,n} \, N_{f,n}^{-1} E_{f,n}$ is \\[-2mm] 
bounded. 
\end{lemma}
\begin{proof}
It suffices to establish the statement for vectors of the special form 
$$
\bPhi_{n-1} = | f_1 \rangle \otimes_s \cdots \otimes_s | f_{n-1} \rangle \,
\in \, \cD_{n-1} \, ,
$$ 
where 
$f_1, \dots, f_{n-1} \in \cD(\RR^s)$ are members of some orthonormal 
basis in $L^2(\RR^s)$ which includes $f$. 
Making use of the fact that the  
Hamiltonians are symmetrized sums of the 
one- and two-particle operators, given above, we obtain 
\begin{align*}
& H_n \, \big(
|f \rangle \otimes_s | f_1 \rangle \otimes_s \cdot \!  \cdot \otimes_s 
| f_{n-1} \rangle \big)  \, - \, |f \rangle \otimes_s H_{n-1} \, 
\big( | f_1 \rangle \otimes_s \cdot \!  \cdot \otimes_s 
| f_{n-1} \rangle \big) \\
& = |\bP_\kappa^2 f \rangle \otimes_s  | f_1 \rangle \otimes_s \cdot \!   
\cdot \otimes_s | f_{n-1} \rangle \!
+ \! \sum_{i=1}^{n-1} \, ( V \, |f \rangle \otimes_s |f_i\rangle )
\otimes_s |f_1 \rangle \otimes_s \cdot \! \cdot \overset{i}{\vee} 
\cdot \! \cdot \otimes_s |f_{n-1} \rangle \, ,
\end{align*}

\vspace*{-2mm}  \noindent 
where the symbol $\overset{i}{\vee}$ indicates the omission of 
the single particle component $|f_i \rangle$. Thus we must determine
the operator on $\cF_n$ which maps the initial vectors 
\mbox{$\,|f \rangle \otimes_s | f_1 \rangle \otimes_s \cdot \cdot \otimes_s 
| f_{n-1} \rangle$} 
to the image vectors on the right hand side of the
preceding equality. Recalling that 
$f, f_1, \dots , f_{n-1}$ are members of some orthonormal basis, we have   
\begin{align*}
(\bP_\kappa^2 E_{f,1} \otimes_s 
\underbrace{1 \otimes_s \cdot \cdot \otimes_s 1}_{n-1} ) \,  \big(
|f \rangle  & \otimes_s | f_1 \rangle \otimes_s \cdots \otimes_s | f_{n-1} 
\rangle \big)  \\[-2mm] 
& = n_f/n \,  
|\bP_\kappa^2 f \rangle \otimes_s | f_1 \rangle \otimes_s \cdots \otimes_s 
| f_{n-1} \rangle  \, ,
\end{align*}

\noindent where $n_f$ is the number of factors $| f \rangle$ appearing in 
the initial vector. This equality holds for arbitrary components $\bPhi_{n-1}$ 
in $| f \rangle \otimes_s \bPhi_{n-1}$ if one replaces 
the number $n_f$ by the operator $N_{f, n}$.
Moreover, since the initial vector is
an element of the space $| f \rangle \otimes_s \cF_{n-1}$, it stays 
constant if
one multiplies it by the projection $E_{f, n}$. This gives     
\begin{align*}
|\bP_\kappa^2 f \rangle & \otimes_s  
| f_1 \rangle \otimes_s \cdots \otimes_s | f_{n-1} \rangle  \\[2mm]
& = n  \, 
(\bP^2 E_{f,1} \otimes_s \underbrace{1 \otimes_s \cdot \cdot \otimes_s 1}_{n-1}) 
 \
N_{f,n}^{-1} E_{f,n} \ 
\big( |f \rangle \otimes_s | f_1 \rangle \otimes_s 
\cdots \otimes_s | f_{n-1} \rangle \big) \\
& = \hat{O}_{f,n} \,  N_{f,n}^{-1} E_{f,n} \,   
\big( |f \rangle \otimes_s | f_1 \rangle \otimes_s 
\cdots \otimes_s | f_{n-1} \rangle \big) \, , 
\end{align*}
where $\hat{O}_{f,n} = \hat{O}_f \upharpoonright \cF_n$ 
and $\hat{O}_f$ is the second quantization of
the one-particle operator $\bP_\kappa^2 E_{f,1}$ on $\cF_1$, having rank one.
So the operator appearing on the right hand side of the second equality 
is the restriction of 
$\hat{A}_f \doteq \hat{O}_f \, N_f^{-1} E_f$ to $\cF_n$, as stated in
the lemma. In a similar manner 
\begin{align*}
& \sum_{i=1}^{n-1} \, (V \, |f \rangle \otimes_s |f_i\rangle ) 
\otimes_s |f_1 \rangle \otimes_s \cdot \cdot \overset{i}{\vee} 
\cdot \cdot \otimes_s |f_{n-1} \rangle \\
& = n(n-1) (V (E_{f,1} \otimes_s 1)  \otimes_s 
\underbrace{1 \otimes_s \cdots \otimes_s 1}_{n - 2} ) \, N_{f,n}^{-1} E_{f,n}
\, (|f \rangle \otimes_s | f_1 \rangle \otimes_s \cdot 
\cdot \otimes_s | f_{n-1} \rangle) \\[-1mm]
& = \hat{V}_{f, n} \,  N_{f,n}^{-1} E_{f,n} \, 
\big( |f \rangle \otimes_s 
| f_1 \rangle \otimes_s \cdot \cdot \otimes_s | f_{n-1} \rangle \big) \, .
\end{align*}
The operator appearing on the right hand side of the second
equality is the restriction $\hat{B}_{f,n}$ of 
$\hat{B}_f \doteq \hat{V}_f \, N_f^{-1} E_f$ to $\cF_n$. 
Since the two-body potential is bounded, 
$\hat{B}_{f,n}$ is bounded, completing the proof. \qed 
\end{proof}

In the last technical lemma of this subsection,
which will also be used further below, we consider  
the adjoint action $\beta_g \doteq \ad \, W_g$ 
on the algebra of bounded operators on $\cF$, which is induced by
the isometries $W_g \in \obfF$ for arbitrary 
normalized $g \in L^2(\RR^s)$, cf.\ Lemma~\ref{l5.1}.
The restrictions of these maps to the algebras 
of bounded operators $\cB(\cF_{n-1})$ 
on $\cF_{n-1}$, having range in $\cB(\cF_n)$, $n \in \NN$, 
are denoted by
\begin{equation} \label{ea.3}
\beta_{g,n} \doteq \beta_g  \upharpoonright \cB(\cF_{n-1}) 
= \ad W_f \upharpoonright \cB(\cF_{n-1}) \, .
\end{equation}
The norm of
arbitrary linear maps $\beta_n : \cB(\cF_{n-1}) \rightarrow 
\cB(\cF_n) $ is denoted by $\| \beta_n \|$, $n \in \NN$. 

\begin{lemma} \label{la.4}
Let $n \in \NN$ and let $g \in L^2(\RR^s)$ be
normalized. Then one has the 
inclusion \mbox{$\beta_{g,n}(\fK_{n-1}) \subset \fK_n$}. 
Moreover, there exists some constant
$c_n$ such that for any pair of normalized elements 
$g_1, g_2 \in L^2(\RR^s)$  
$$
\| \beta_{g_1,n} - \beta_{g_2,n} \| \leq c_n \, \|g_1 - g_2 \|_2  \, .
$$
\end{lemma}
\begin{proof}
According to \cite[Lemma~3.3]{Bu1} there exists for given $K_{n-1} \in \fK_{n-1}$
some observable $A_{n-1} \in \obfA$ such that 
$A_{n-1} \upharpoonright \cF_{n-1} = K_{n-1}$. It follows from Lemma~\ref{l5.1}
and \ref{l3.2} that 
$W_g A_{n-1} W_g^* \in \obfA$. Hence, applying 
the results in \cite[Lem.~3.3]{Bu1} another time, one obtains 
$$
\beta_{g,n}(K_{n-1}) = 
(W_g A_{n-1} W_g^*) \upharpoonright \cF_n \, \in \,  
\obfA \upharpoonright \cF_n = \fK_n \, ,
$$
as claimed. The continuity of the maps 
is a consequence of Lemma~\ref{l5.1}, which leads to the 
estimate 
$$
\| \beta_{g_1,n} - \beta_{g_2,n} \|
\leq 2 \, \| \big( W_{g_1}^* - W_{g_2}^* \big) P_n \|  
\leq c_n \, \|g_1 - g_2\|_2 \, ,
$$
completing the proof. \qed \end{proof}

We have gathered now the information needed for the description of 
the structure of the operator 
$\big(H - \beta_f(H)\big) \, E_f$. 

\begin{proposition} \label{pa.5} 
Let $n \in \NN_0$, \ then 
$$ \big(H - \beta_f(H)\big) \, E_f \upharpoonright \cF_n
=  A_{f, n} + B_{f, n} \, .
$$
Here $A_{f, n} = A_f \upharpoonright \cF_n$, where
$ A_f = \hat{A}_f + \beta_f \, \scirc \, \sigma_f^{-1}(\check{A}_f)$
and the operators $\check{A}_f$, 
$\hat{A}_f$ were defined in Lemmas~\ref{la.2} and \ref{la.3}, respectively.  
One has $A_{f, n} \in \fK_n$. 
In a similar manner, $B_{f, n} = B_f \upharpoonright \cF_n$, where
$B_f = \hat{B}_f + \beta_f \, \scirc \, \sigma_f^{-1}(\check{B}_f)$    
and the operators $\check{B}_f$,~$\hat{B}_f$ were likewise defined in these 
two lemmas. The operator $B_{f, n}$ is bounded.  
\end{proposition} 
\begin{proof}
Recalling that $E_f \cF_n = | f \rangle \otimes_s \cF_{n-1}$, one obtains
for $\bPhi_{n-1} \in \cD_{n-1}$
\begin{align*}
\big(H_n - \beta_f(H_{n-1})\big) \, ( | f\rangle \otimes_s \bPhi_{n-1} ) 
& = \big(H_n \, (|f \rangle \otimes_s \bPhi_{n-1}) - 
|f \rangle \otimes_s  H_{n-1} \, \bPhi_{n-1} \big) \\
& + |f\rangle\otimes_s\big(H_{n-1}-\sigma_{f,n-1}(H_{n-1})\big) \bPhi_{n-1} \, .
\end{align*}
The first term on the right hand side of this equality coincides  
according to Lemma~\ref{la.3} 
with $(\hat{A}_{f,n} + \hat{B}_{f,n}) \, ( | f \rangle \otimes_s \bPhi_{n-1})$, 
where $\hat{A}_{f,n} \in \fK_n$ and
$\hat{B}_{f,n}$ is bounded. \  
In the second term we made use of Lemma~\ref{la.1}(i) according to which
$
\beta_f(H_{n-1}) \ \big( | f\rangle \otimes_s \bPhi_{n-1} \big) 
=  | f\rangle \otimes_s \sigma_{f,n-1}(H_{n-1}) \, \bPhi_{n-1}  \, .
$
Hence this second term can be presented in the 
form \ $| f \rangle \otimes_s (\check{A}_{f,n-1} + \check{B}_{f,n-1}) 
\bPhi_{n-1}$, as was shown in Lemma~\ref{la.2}. According to 
Lemma~\ref{la.1}(ii),  the latter  vector coincides with 
the image of $| f \rangle \otimes \bPhi_{n-1}$ under the action of 
$\beta_f \, \scirc \, \sigma_f^{-1}(\check{A}_f + \check{B}_f) 
\upharpoonright \cF_n$.

\medskip
Turning to the proof that $A_{f,n} \in \fK_n$, we note that 
$$
\sigma_{f,n-1}^{-1}(\check{A}_{f,n-1}) = (\one_{n-1} + N_{f,n-1})^{1/2}
\check{A}_{f,n-1} (\one_{n-1} + N_{f,n-1})^{-1/2} \in \fK_{n-1} \, .
$$
It therefore follows from the preceding lemma that
$$
\beta_f \, \scirc \, \sigma_f^{-1}(\check{A}_f) \upharpoonright \cF_n
= \beta_{f,n} \, \scirc \, \sigma_{f,n-1}^{-1}(\check{A}_{f,n-1}) \in \fK_n \, .
$$ 
Since also $\hat{A}_{f,n} = \hat{A}_f \upharpoonright \cF_n \in \fK_n$, 
we obtain  $A_{f,n} \in \fK_n$. That $B_{f,n}$ is bounded is 
apparent, completing the proof. \qed
\end{proof}

\subsection{Dyson expansions with values in $\fK_n$} \label{ssa.3} \hfill

\noindent We turn now to the analysis of the operator function 
$t \mapsto \Gamma_f(t) E_f$, 
defined in Eqn.~\eqref{e5.2}. It is differentiable in $t$ in the sense of
sesquilinear forms between vectors in the domains of $H$, respectively
$W_f H W_f^*$. The derivatives are given by
\begin{align*}
\frac{d}{dt} \, \Gamma_f(t) \, E_f & =
i \, e^{itH}(H - W_f H W_f^*) \, e^{-it W_f H W_f^*} \, E_f \\
& = i \, e^{itH}(H - W_f H W_f^*) E_f \, e^{-it W_f H W_f^*} \, E_f \\
& = i \, e^{itH}(H - W_f H W_f^*) E_f \, e^{-itH} \ \Gamma_f(t) \, E_f \, ,
\end{align*}
where the second equality holds since   
$W_f H W_f^*$ commutes with~$E_f$. 
We restrict this equality to $\cF_n$ and
put $\Gamma_{f,n}(t) \doteq \Gamma_f(t) \, E_f \upharpoonright \cF_n$,
$n \in \NN_0$. In particular, 
$\Gamma_{f,n}(t) \upharpoonright (\one_n - E_{f,n}) \, \cF_n = 0$. 
By Proposition~\ref{pa.5} we have  
$$
(H - W_f H W_f^*) E_f \upharpoonright \cF_n = A_{f,n} + B_{f,n} \, ,
$$ 
where $A_{f,n} \in \cK_n$ and $B_{f,n}$ is a bounded operator. 
Putting 
$$
C_{f,n}(s) \doteq \ad e^{isH_n} (A_{f,n} + B_{f,n}) \, ,
\quad s \in \RR \, ,
$$
we can solve the above differential equation for
$t \mapsto \Gamma_{f,n}(t)$  on $E_{f,n} \, \cF_n$ by the Dyson series
of time ordered integrals, defined in the strong
operator topology,
\begin{equation} \label{ea.4}
\Gamma_{f,n}(t) 
= \big(E_{f,n} + \sum_{k=1}^\infty i^k \,
\int_0^t \! ds_k  \! \int_0^{s_k} \! ds_{k-1}  \dots  \! \int_0^{s_2} \! ds_1 \,
C_{f,n}(s_k) \cdots  C_{f,n}(s_1) \big) \, .  
\end{equation}
This series converges absolutely in norm since the operators $C_{f,n}$ are
bounded. 

\medskip 
We want to show that 
$\Gamma_{f,n}(t) \in \fK_n$, $n \in \NN_0$.
As we shall see, it is sufficient to prove that the functions \ 
$t \mapsto \int_0^t \! ds \, C_{f,n}(s)$ have range in 
$\fK_n$ and are norm continuous, $t \in \RR$. For the summand 
$A_{f,n} \in \fK_n$ of $C_{f,n}$ this property 
follows from the fact that the time evolution acts pointwise norm 
continuously on $\fK_n$, cf.\ Proposition~4.4 and the 
appendix in~\cite{Bu1}. 
The argument for the second summand~$B_{f,n}$ is more involved since
these operators are not contained in $\fK_n$. We begin with a 
technical lemma about integrals of functions having  
values in operators, respectively linear maps between 
C*-algebras. In order to avoid repetitions of technicalities, 
we make the following standing assertion.

\medskip 
\noindent \textbf{Statement:}  In the subsequent analysis all
integrals are defined in the strong operator (s.o.) topology of the 
underlying Hilbert spaces, unless otherwise stated.

\begin{lemma} \label{la.6}
For $k = 1,2$, let $\cH_k$ be a Hilbert space and  
let $\fB_k \subset \cB(\cH_k)$ be a C*-algebra. 
Moreover, let $s \mapsto B_1(s) \in \cB(\cH_1)$, $s \in \RR$,  be 
a s.o. continuous operator function such that 
$\int_0^t \! ds \,  B_1(s) \in \fB_1$, $t \in \RR$; 
and let $s \mapsto \lambda(s)$ be a norm continuous 
function with values in linear~maps from $\cB(\cH_1)$ into
$\cB(\cH_2)$, which, for fixed $s \in \RR$, are 
normal (s.o. continuous)  
and whose restrictions to $\fB_1$ have values in $\fB_2$, \ie
$\lambda(s)(\fB_1) \subset \fB_2$.

\medskip 
Then the function $s \mapsto \lambda(s)\big(B_1(s)\big) \in \cB(\cH_2)$ is 
s.o.\ continuous. Its integral 
$\, t \mapsto \int_0^t \! ds \, \lambda(s)\big(B_1(s)\big)$ is 
norm continuous and has values in 
$\fB_2$, $t \in \RR$. For fixed~$t$, 
it can be approximated in norm in the limit
$m \rightarrow \infty$ by the sums
$$
\sum_{l = 1}^m \lambda(lt/m)\Big(\int_{(l-1)t/m}^{lt/m} \! \! ds \, 
B_1(s) \Big) \in \fB_2 \, , \quad m \in \NN \, .
$$
(Note that the functions 
in this lemma are not necessarily defined by the action of some dynamics.)
\end{lemma}

\begin{proof}
Let $s_0 \in \RR$. Then one has on $\cH_2$ the equality  
$$
\lambda(s) \big(B_1(s) \big) - \lambda(s_0) \big(B_1(s_0) \big) =
\lambda(s_0) \big(B_1(s) - B_1(s_0) \big) + 
\big(\lambda(s) - \lambda(s_0)\big) \big(B_1(s) \big) \, .
$$
Since the map $\lambda(s_0)$ is normal on
$\cB(\cH_1)$, the first term on the right hand side
of this equality vanishes in the s.o.\ topology in the limit    
$s \rightarrow s_0$. The second term vanishes in this limit as well, 
since $\lambda(s) \rightarrow \lambda(s_0)$ in the norm topology 
of $\cB(\cH_2)$, uniformly on 
bounded subsets of $\cB(\cH_1)$. Thus $s \mapsto \lambda(s) \big(B_1(s) \big)$ 
is continuous in the s.o.\ topology and the integrals exist.
Assuming without loss of generality that $t \geq 0$,
we partition $[0,t]$ into $m \in \NN$ intervals, giving the estimate 
in~$\cB(\cH_2)$ 
\begin{align*}
 \| \int_0^t \! ds \, \lambda(s) & \big(B_1(s)\big) 
  - \sum_{l = 1}^m \lambda(lt/m)\Big(\int_{(l-1)t/m}^{lt/m} \! \! ds \, 
B_1(s) \Big) \| \\
& = \|  \sum_{l = 1}^m  
\int_{(l-1)t/m}^{lt/m} \! ds \, 
\big(\lambda(s) - \lambda(lt/m) \big) \big(B_1(s)\big) \| \\
& \leq \| B_1 \|_\infty \sum_{l = 1}^m  \int_{(l-1)t/m}^{lt/m} \! ds \,
 \| \lambda(s) - \lambda(lt/m) \| \, ,
\end{align*}
where $\| B_1 \|_\infty \doteq \sup_{0 \leq s \leq t} \| B(s) \|$.
Because of the norm continuity of 
$s \mapsto \lambda(s)$, 
this shows that the expression on the first line tends to $0$ in the
limit $m \rightarrow \infty$. Since, by assumption, 
$\int_{(l-1)t/m)}^{lt/m} \! ds \, B_1(s) \in \fB_1$ and
$\lambda(lt/m)$ maps the C*-algebra $\fB_1$ into $\fB_2$, $1 \leq l \leq m$, 
it follows that 
$\int_0^t \! ds \, \lambda(s)\big(B_1(s)\big) \in \fB_2$.
Moreover, the integral can be approximated in norm by the finite sums
given in the lemma.

\medskip
The statement about the 
continuity properties follows from the estimate 
\begin{align*}
\| \int_0^{t_2} \! ds \, \lambda(s) & \big(B(s) \big) - \int_0^{t_1} \! ds \, 
\lambda(s) \big(B(s) \big) \| \\
& = \| \int_{t_1}^{t_2} \! ds \, \lambda(s) \big(B(s) \big) \|
\leq \| B \|_\infty \|\lambda\|_\infty |t_2 - t_1| \, ,
\end{align*}
where $\| B_1 \|_\infty$, $\|\lambda\|_\infty$
are the suprema of $s \mapsto \| B_1(s) \|$,
respectively \mbox{$s \mapsto \| \lambda(s) \|$}, 
on any given bounded subset of~$\RR$, containing the 
integration intervals. 
\qed \end{proof} 

This lemma will be applied to various types of functions 
and has therefore been formulated in general terms. It will 
allow us to determine the properties of integrals
involving the localized potentials, cf.~\eqref{ea.4}. In the subsequent 
lemma we consider integrals of operators evolving under 
the non-interacting time evolution $e^{isH_0}$, $s \in \RR$. Here 
$H_0$ is the second quantization of the single particle 
operator $\bP_\kappa^2 = \bP^2 + \kappa^2 \, \bQ^2$, where 
$\kappa \geq 0$ is kept fixed. We make use of the short
hand notation $B(s)_{0} = \ad e^{isH_0}(B)$ for arbitrary
bounded operators $B$ on $\cF$. If $B$ is gauge invariant 
(preserves the particle number), we denote its 
restriction to the $n$-particle space by 
$B_n(s)_{0} \doteq \ad e^{isH_0,n}(B_n)
= B(s)_{0} \upharpoonright \cF_n$, $n \in \NN_0$.
Note that the functions $s \mapsto B(s)_{0}$ are 
continuous in the s.o. topology.

\begin{lemma} \label{la.7}
Let $B_n \in \cB(\cF_n)$ such that the functions 
$t \mapsto \int_0^t \! ds \, B_n(s)_0$, $t \in \RR$, have values in $\fK_n$,
$n \in \NN_0$. 

\medskip 
(i) The functions 
$t \mapsto \int_0^t \! ds \, \big(K_n' B_n K_n'' \big)(s)_0$ 
have values in $\fK_n$ and are norm continuous for any
choice of $K_n', K_n'' \in \fK_n$. 

\medskip
(ii) Let $\beta_{g,n} : \cB(\cF_{n-1}) \rightarrow \cB(\cF_n)$ be the map  
defined in Eqn.~\eqref{ea.3}  
for normalized $g \in L^2(\RR^s)$. The functions 
$t \mapsto \int_0^t \! ds \, \big(\beta_{g,n}(K_{n-1}' B_{n-1} K_{n-1}'') 
\big)(s)_0$ have values in $\fK_n$ and are norm continuous for any choice of 
$K_{n-1}', K_{n-1}'' \in \fK_{n-1}$.
\end{lemma}
\begin{proof}
(i) Consider the function 
$\lambda_n : \RR \times \cB(\cF_n) \rightarrow
\cB(\cF_n)$, which is given by
$$ \lambda_n(s)(B_n') \doteq K_n'(s)_0 \, B_n' \, K_n''(s)_0
\, ,  \quad s \in \RR \, , \ B_n' \in \cB(\cF_n) \, .
$$
For fixed $s$ the linear 
map $\lambda_n(s)$ is clearly normal. Moreover, since 
the time translations leave the algebra of 
observables $\obfA$ invariant and act pointwise norm continuously on 
$\fK_n = \obfA \upharpoonright \cF_n$, cf.\ Proposition~4.4 and the
appendix in \cite{Bu1}, 
the function is norm continuous and its restriction to $\fK_n$
maps this subalgebra into itself. The function 
$s \mapsto B_n(s)_0$ is continuous in the s.o.\ topology
and, by assumption, $t \mapsto \int_0^t \! ds \, B_n(s)_0 \in \fK_n$, 
$t \in \RR$, 
so the first statement follows from Lemma~\ref{la.6}.

\medskip
(ii) Since we are dealing with the non-interacting time evolution, 
we have 
$$
\big(\beta_{g,n}(K_{n-1}' B_{n-1} K_{n-1}'') \big)(s)_0
= \beta_{g(s)_0 ,n} \, (K_{n-1}'(s)_0  \, B_{n-1}(s)_0 \, K_{n-1}''(s)_0 )  \, ,
$$
where $g(s)_0 \doteq e^{is \sbP_\kappa^2} \, g \in L^2(\RR^s)$ is 
normalized. We consider now the 
function of maps $\lambda_n : \RR \times \cB(\cF_{n-1}) \rightarrow
\cB(\cF_n)$, given by
$$ \lambda_n(s)(B_{n-1}') \doteq 
\beta_{g(s)_0 ,n}(K_{n-1}'(s)_0 \,  B_{n-1}' \, K_{n-1}''(s)_0 ) \, ,
\ \
s \in \RR \, , \ B_{n-1}' \in \cB(\cF_{n-1}) \, .
$$
For fixed $s$ the linear map $\lambda_n(s)$ is normal.
Moreover, it follows from the norm continuity of 
$s \mapsto \beta_{g(s)_0 ,n}$, established in 
the second part of Lemma~\ref{la.4},  
and the arguments in step (i) that $s \mapsto \lambda_n(s)$ is 
norm continuous. Making use of the first part of 
Lemma~\ref{la.4}, it is also clear that $\lambda_n(s)$ maps
$\fK_{n-1}$ into $\fK_n$, $s \in \RR$.  
Furthermore, the function 
$s \mapsto B_{n-1}(s)_0$ is continuous in the s.o.\ topology
and, by assumption, $t \mapsto \int_0^t \! ds \, B_{n-1} (s)_0 \in \fK_{n-1}$, 
$t \in \RR$.
So the second  statement follows likewise from Lemma~\ref{la.6}.
\qed \end{proof}

In the next lemma, we consider the non-interacting time evolution 
of localized pair potentials,
defined above, and study their integrals. 

\begin{lemma} \label{la.8}
Let $\check{V}_{f,2}$ and $\hat{V}_{f,2}$ be the localized pair potentials, 
defined in Lemmas~\ref{la.2} and \ref{la.3}, respectively.
Denoting by $V_{f,2}$ either one of these potentials, one has 

\medskip 
(i) the function 
$t \mapsto \int_0^t \! ds \, V_{f,2}(s)_0$  on $\cF_2$
is norm continuous and has values in the compact operators;

\medskip 
(ii) for any $n \in \NN$, $n \geq 2$, the function 
\ $t \mapsto \int_0^t \! ds \, V_{f, n}(s)_0$ on $\cF_n$ 
is norm continuous and 
has values in $\fK_n$, $t \in \RR$, where 
$$ 
s \mapsto V_{f,n}(s)_0 = 
n(n-1) \, (V_{f,2}(s)_0 \otimes_s 
\underbrace{1 \otimes_s \cdots \otimes_s 1}_{n-2}) \, . 
$$
\end{lemma}
\begin{proof} 
We give the proof for the potential 
$\hat{V}_{f,2} = V \, (E_{f,1} \otimes_s 1)$. Since  
$\check{V}_{f,2}$ also contains the localizing factor $(E_{f,1} \otimes_s 1)$,
the corresponding argument is similar and therefore omitted. 

\medskip 
(i) In a first step we consider potentials $V$, having compact support. 
Picking a smooth function $\bx \mapsto \chi(\bx)$ 
which is equal to $1$ for 
$\bx \in \text{supp} \, f \cup (\text{supp} \, f + \text{supp} \, V)$
and has compact support, we can proceed to  
$\hat{V}_{f,2} = V_{f, \chi} \, (E_{f,1} \otimes_s 1)$, 
where the potential 
$\bx, \by \mapsto V_{f, \chi}(\bx,\by) \doteq 
V(\bx - \by) \, \chi(\bx) \chi(\by)$ 
is symmetric in $\bx, \by$, continuous, and compactly supported 
on the two-particle configuration
space $\RR^s \times \RR^s$. The function $s \mapsto V_\chi(s)_0$
is continuous in the s.o.\ topology and the resulting integral 
$t \mapsto \int_0^t \! ds \, V_\chi(s)_0$ 
depends norm continuously on $t \in \RR$ and has values 
in the compact operators on $\cF_2$. Similar results were 
established in previous work, cf. for example the appendix
of \cite{BuGr1}. We briefly sketch here the argument for the 
case at hand. Consider the functions, $\kappa \geq 0$, 
$$
s \mapsto c_\kappa(s) \doteq \cos(2\kappa s) \, , \quad
s \mapsto s_\kappa(s) \doteq \sin(2\kappa s)/\kappa \, \quad s \in \RR \, ,
$$
where we put $s_0(s) = 2s$. The non-interacting time translations
act on the two-particle operator 
$V_{f, \chi} = V_{f, \chi}(\bQ_1, \bQ_2)$ according
to
$$
V_{f, \chi}(\bQ_1,\bQ_2)(s)_0 =
V_{f, \chi}(c_\kappa(s) \bQ_1 + s_\kappa(s) \bP_1, \ 
c_\kappa(s) \bQ_2 + s_\kappa(s) \bP_2) 
$$
in an obvious notation. For any $(s', s'') \in \RR^2$ the operators
$\big( c_\kappa(s') \bQ_1 + s_\kappa(s') \bP_1 \big) $ commute with 
$\big( c_\kappa(s'') \bQ_2 + s_\kappa(s'') \bP_2 \big)$ and one has,
$l = 1,2$,  
$$
[ \big( c_\kappa(s') \bQ_l + s_\kappa(s') \bP_l \big),
\big( c_\kappa(s'') \bQ_l + s_\kappa(s'') \bP_l \big)] 
=  i s_\kappa(s'' - s') \, \one 
\, . 
$$
Thus for almost all $(s', s'') \in \RR^2$ the operators in the 
latter commutator do not commute and 
are canonically conjugate (with rescaled Planck constant). 
Since $\bx,\by \mapsto V_{f,\chi}(\bx,\by)$ is continuous  
and has compact support it follows
from standard arguments that the function 
$$
s', s'' \mapsto V_{f, \chi}(\bQ_1,\bQ_2)(s')_0 \, V_{f, \chi}(\bQ_1,\bQ_2)(s'')_0
$$
has values in compact operators on $\cF_2$ 
for almost all $(s', s'') \in \RR^2$.
Since it is also uniformly bounded, this implies that 
$$
| \int_0^t \! ds \,  V_{f, \chi}(\bQ_1,\bQ_2)(s)_0 |^2 =
\int_0^t \! ds'  \!  \! \int_0^t \! ds''
 V_{f, \chi}(\bQ_1,\bQ_2)(s')_0 \, V_{f, \chi}(\bQ_1,\bQ_2)(s'')_0 
$$
is a compact operator. Taking its square root and making use of 
polar decomposition, we conclude that 
$\int_0^t \! ds \,  
V_{f, \chi}(s)_0 = \int_0^t \! ds \,  V_{f, \chi}(\bQ_1,\bQ_2)(s)_0 $
is compact on $\cF_2$. (It is this result where we made use of the special 
form of the external single particle potential; but we expect that it 
holds more generally.)

\medskip
In order to see that this conclusion holds also for the original 
localized potential $V_f = V (E_{f,1} \otimes_s 1)$, we introduce the function 
$\lambda : \RR \times \cB(\cF_2) \rightarrow \cB(\cF_2)$ given by
$$
\lambda(s)(B_2) \doteq B_2 \, 
(E_{f,1} \otimes_s 1)(s)_0 =  B_2 \, 
(E_{f,1}(s)_0 \otimes_s 1) \, ,   \quad s \in \RR \, , \, 
B_2 \in \cB(\cF_2) \, .
$$
It is linear, normal, continuous in norm (recall that $E_{f,1}$ is a
one-dimensional projection), and it maps compact operators on 
$\cF_2$ to compact operators. It therefore follows from 
Lemma~\ref{la.6} that 
$$ 
t \mapsto
\int_0^t \! ds \, \lambda(s)(V_{f,\chi}(s)_0) =  \!
\int_0^t \! ds \, \big(V_{f,\chi} \, (E_{f,1} \otimes_s 1)\big)(s)_0  
=  \! \int_0^t \! ds \, \big(V \, (E_{f,1} \otimes_s 1)\big)(s)_0 
$$
is norm continuous and has values in 
the compact operators on $\cF_2$ for pair potentials with
compact support. The last integral in the 
preceding equality is norm continuous on $\cF_2$
with regard to $V \in C_0(\RR^s)$, equipped with 
the supremum topology. Since the algebra of compact operators
is norm-closed, the preceding result for pair potentials with
compact support therefore extends to all 
potentials in $C_0(\RR^s)$. 

\medskip 
(ii) By the very definition of the spaces $\cK_n$,  
any compact operator $C$ on $\cF_2$ gives rise to elements 
$C \otimes_s \underbrace{1 \otimes \cdots \otimes 1}_{n-2} \in
\cK_n$, $n \in \NN$. Moreover, for the non-interacting dynamics
one has 
$$
\big( C \otimes_s \underbrace{1 \otimes \cdots \otimes 1}_{n-2} \big) (s)_0
= C(s)_0 \otimes_s \underbrace{1 \otimes \cdots \otimes 1}_{n-2} \, .
$$ 
So the second statement follows from the
preceding one. As has been mentioned, analogous arguments
apply to the localized pair potential $\check{V}_{f,2}$, 
completing the proof.  
\qed \end{proof}

\medskip
Next, we show that the assumptions in 
Lemma~\ref{la.7} imply that the statements
(i) and (ii) still hold if one replaces the non-interacting 
time evolution in the respective integrals by the interacting one.
This fact will enable us to show that the 
functions $t \mapsto \int_0^t \! ds \, C_{f,n}(s)$ in the 
Dyson expansion \eqref{ea.4} are norm continuous
and have values in $\fK_n$, $n \in \NN$. We recall the short hand notation  
$B_n(s) \doteq \ad e^{isH_{n}} (B_n)$ for the adjoint action of the
interacting dynamics and $B_n(s)_0 \doteq \ad e^{isH_{0,n}} (B_n)$
in case of no interaction, $n \in \NN_0$.

\begin{lemma} \label{la.9}
Let $n \in \NN_0$ and let $B_n \in \cB(\cF_n)$ be an operator 
such that the function 
$t \mapsto \int_0^t \! ds \, B_n(s)_0$, 
defined in terms of the non-interacting time evolution, has values in $\fK_n$.
Then the function $t \mapsto \int_0^t \! ds \, B_n(s)$,
involving the interacting time evolution, has values in $\fK_n$
and is norm continuous, $t \in \RR$.
\end{lemma}
\begin{proof}
Let $\Lambda_n(s)^{-1} \doteq  e^{isH_{0,n}} e^{-isH_n}$
and put $\lambda_n(s)^{-1}  \doteq \ad  
\Lambda_n(s)^{-1}$, $s \in \RR$. Given $B_n \in \cB(\cF_n)$, one obtains by the
familiar Dyson expansion the equality
\begin{align*}
& \lambda_n(s)^{-1}(B_n) = B_n \\
& + \! 
\sum_{k = 1}^\infty (-i)^k \! \! \int_0^s \! \!  du_k \! \! \int_0^{u_k}  
\! \! \! du_{k-1} \dots
\! \! \int_0^{u_2} \! \! \! \! \! du_1 \,
[V_{\! n}(u_k)_0, [V_{\!n}(u_{k-1})_0, [ \cdots [V_{\!n}(u_1)_0, B_n]]] \cdots ]
\hspace{0,2mm} ,
\end{align*}
where  $V_{\!n}$ is the restriction of the interaction potential to $\cF_n$.
Since the underlying pair potential is bounded, this series converges 
absolutely in the norm topology. Moreover, for $s_2 \geq s_1$ 
one has
$$
\| (\lambda_n(s_2)^{-1} - \lambda_n(s_1)^{-1})(B_n) \| 
\leq \| B_n \| \, 
\sum_{k=1}^\infty 2^k/k! \ \| V_{\!n} \|^k \, \int_{s_1}^{s_2} 
\! du \, |u|^{k-1} \, ,
$$
proving that the function $s \mapsto \lambda_n(s)^{-1}$ of linear maps 
on $\cB(\cH_n)$ is norm continuous. 
It is also apparent from the dominated convergence theorem 
that these maps are normal for fixed $s$. Finally, it
was shown in Proposition~4.4 and the appendix of \cite{Bu1} that 
the interacting and non-interacting time evolutions map the 
algebra $\fK_n$ onto itself, hence this is also true for 
$\lambda_n(s)^{-1}$, $s \in \RR$. 

\medskip
Thus the maps $\lambda_n(s)^{-1}$ are automorphisms, both, of $\cB(\cF_n)$ and
of~$\fK_n$. So all preceding statements hold also for
their inverse $\lambda_n(s)$, given by the adjoint action of 
$\Lambda_n(s) = e^{isH_n} e^{-isH_{0, n}}$, $s \in \RR$. Hence
the maps $s \mapsto \lambda_n(s)$ comply with all 
conditions given in Lemma~\ref{la.6}. Moroever, the 
function $s \mapsto B_n(s)_0$ is s.o. continuous and, by 
assumption, $\int_0^t \! ds  \,  B_n(s)_0 \in \fK_n$, $t \in \RR$. It therefore
follows from Lemma~\ref{la.6} that 
$$
t \mapsto \int_0^t \! ds  \,  B_n(s) = \int_0^t \! ds  \, 
\lambda_n(s)(B_n(s)_0) 
$$
is norm continuous and has values in $\fK_n$,
completing the proof. \qed \end{proof}

\medskip
With the help of the preceding three lemmas we can establish now the 
main result of this subsection.

\begin{proposition} \label{pa.10}
Let $n \in \NN_0$.  The function 
$t \mapsto \Gamma_{f,n}(t)$, given in Eqn.~\eqref{ea.4}, 
is norm-continuous
and has values in $\fK_n$, $t \in \RR$.
\end{proposition}
\begin{proof}
We begin by studying the properties of the function 
$t \mapsto \int_0^t \! ds \, C_{f,n}(s)$, which appears in lowest non-trivial
order of the series expansion \eqref{ea.4}.
According to  Proposition~\ref{pa.5} one has   
\mbox{$C_{f,n} = A_{f,n} + B_{f,n}$}, where $A_{f,n} \in \fK_n$
and $B_{f,n}$ is bounded. 

\medskip
It was shown in  
\cite[Prop.~4.4]{Bu1} and the appendix of that   
reference that the function $ s \mapsto A_{f,n}(s)$, involving the
interacting dynamics, is norm continuous, $s \in \RR$.
Its integrals $t \mapsto \int_0^t \! ds \, A_{f,n}(s)$ are therefore 
defined in the norm topology, whence have values in $\fK_n$, and
depend norm continuously on $t \in \RR$.

\medskip
Turning to the operators $B_{f,n}$, we recall their form
\ $B_{f} = \hat{B}_{f} + \beta_{f} \, \scirc \, 
\sigma_{f}^{-1}(\check{B}_{f})$, established in 
Proposition~\ref{pa.5}. Plugging into this
equation the operators $\check{B}_f, \hat{B}_f$, given 
in Lemmas~\ref{la.2} and \ref{la.3}, as well as 
the maps $\sigma_f$, $\beta_f$,
defined in Eqns.~\eqref{ea.2} and \eqref{ea.3},
respectively, we obtain for $B_{f,n} = B_f \upharpoonright \cF_n$, $n \in \NN$,
$$
B_{f,n} = \hat{V}_{f,n} N_{f,n}^{-1} E_{f,n}
+ \beta_{f,n} \big((1 + N_{f,n-1})^{1/2} \check{V}_{f,n-1}
(1 + N_{f,n-1})^{-1/2} - \check{V}_{f,n-1} \big) \, .
$$
It was shown in Lemma~\ref{la.8} that the 
integrals of the localized potentials  
with regard to the non-interacting dynamics, 
$\int_0^t \! ds \, \hat{V}_{f,n}(s)_0$, \
$\int_0^t \! ds \, \check{V}_{f,n}(s)_0$, are elements of $\fK_n$, 
$t \in \RR$. Since the operators $N_{f,n}^{-1} E_{f,n}$ and
$(1 + N_{f,n})^{\pm 1/2}$ are elements of $\fK_n$,
Lemma~\ref{la.7} implies that this is also  
true for the integral $\int_0^t \! ds \, B_{f,n}(s)_0$. 
It then follows from Lemma~\ref{la.9}
that the integral with regard to the interacting 
dynamics, $\int_0^t \! ds \, B_{f,n}(s)$ has values 
in $\fK_n$. Combining the preceding results, we see that the 
function $t \mapsto \int_0^t \! ds \, C_{f,n}(s)$ 
has values in $\fK_n$. Since $C_{f,n}$ is bounded,
it is also clear that it is norm continuous, $t \in \RR$.

\medskip 
The proof that also the contributions of higher order in the 
series expansion~\eqref{ea.4} are contained in $\fK_n$ 
is accomplished by induction. Putting 
\begin{equation} \label{ea.5}
t \mapsto D_{n,k}(t) \doteq 
\int_0^t \! \! ds_k \! \! \int_0^{s_k} \! \!  \! ds_{k-1} \! \dots \! \!
\int_0^{s_2} \! \! ds_1 \,
C_{f,n}(s_k) \cdots  C_{f,n}(s_1) \, ,  \ \ k \in \NN \, ,
\end{equation}
we will show that these functions have values in $\cK_n$ and
are norm continuous, $t \in \RR$.
For $k=1$ this was shown in the preceding step.

\medskip 
For the induction step from $k$ to $k+1$, we make use of the fact that
relation~\eqref{ea.5} implies that \   
$D_{n,k+1}(t) = \int_0^t \! ds \, C_{f,n}(s) \, D_{n,k}(s)$. 
According to the induction hypothesis, the function 
$s \mapsto D_{n,k}(s)$ is norm  continuous and has values in~$\fK_n$.
Thus the function $s \mapsto \lambda_{k,n}(s)$ of normal linear maps
on $\cB(\cF_n)$ given by 
$\lambda_{k,n}(s)(B_n) \doteq B_n D_{n,k}(s)$, $B_n \in \cB(\cF_n)$,
is norm continuous and maps $\fK_n$ into itself.
The function $s \mapsto C_{f,n}(s)$ 
is s.o. continuous and $\int_0^t \! ds \, C_{f,n}(s) \in \fK_n$,
$t \in \RR$, as was shown in the initial step. Hence, according to 
Lemma~\ref{la.6}, the function 
$t \mapsto D_{n,k+1}(t) = \int_0^t \! ds \, C_{f,n}(s) \, D_{n,k}(s) = 
\int_0^t \! ds \, \lambda_{k,n}(s)\big(C_{f,n}(s)\big)$ has the
desired properties, completing the induction.

\medskip
So each term in the Dyson expansion \ref{ea.4}
is an element of $\fK_n$. Moreover, this series converges absolutely
in the norm topology, which implies $\Gamma_{f,n}(t) \in \fK_n$.
Since the operators $C_{f,n}$ are bounded, it is also
clear that the function $t \mapsto \Gamma_{f,n}(t)$ is norm continuous,
$t \in \RR$, 
cf.\ the argument in Lemma \ref{la.9}. This completes the
proof of the proposition.
\qed \end{proof}

\subsection{Verification of the coherence condition} \label{ssa.4} \hfill

\noindent Having seen that the operators 
$\Gamma_{f,n}(t)$, defined in Eqn.~\eqref{ea.4},
are elements of $\cK_n$, $t \in \RR$, 
we will show next that these operators form coherent sequences,
$n \in \NN_0$. 
At this point the inverse maps $\kappa_n: \cK_n \rightarrow \cK_{n-1}$,
defined in Eqn.~\eqref{e3.2}, 
enter. We recall that these maps are homomorphisms,
mapping $\cK_n$ onto $\cK_{n-1}$, and that a sequence of operators 
$\bK = \{ K_n \in \fK_n \}_{n \in \NN_0}$ is said to be coherent if
$\kappa_n(K_n) = K_{n-1}$, $n \in \NN_0$.

In order to establish the desired result, we make use again of the
Eqn.~\eqref{ea.5}, relating subsequent terms in the Dyson
expansion \eqref{ea.4}. The essential step in our argument 
consists of proving the equality \ 
$$ 
\kappa_n \Big( \int_0^t \! ds \, C_{f,n}(s) D_n(s) \Big) = 
\int_0^t \! ds \, C_{f,n-1}(s) \, \kappa_n\big(D_n(s)\big) 
$$
for any norm continuous function $s \mapsto D_n(s)$ with values in 
$\fK_n$, $n \in \NN_0$. Since the values of the functions 
$s \mapsto C_{f,n}(s)$ 
are not contained in $\fK_n$, this requires some
further arguments. We begin with a statement, involving
the non-interacting time evolution. 

\begin{lemma} \label{la.11} 
Let $m = 1$ or $2$, let $O_m$ be a bounded $m$-particle operator
on $\cF_m$ such that $\ \int_0^t \! ds \, O_m(s)_0 \, $ is a compact 
operator, $t \in \RR$, and let $O$ be the second quantization of $O_m$.  
Putting $O_n \doteq O \upharpoonright \cF_n$, one has 
$\, \int_0^t \! ds \, O_n(s)_0  \in \fK_n \, $, $t \in \RR$, and 
$$ \kappa_n \Big( \int_0^t \! ds \, O_n(s)_0 \Big) = 
\int_0^t \! ds \, O_{n-1}(s)_0 \, , \quad n \in \NN_0 \, . $$
\end{lemma}
\begin{proof}
The second quantizations of one- and two-particle operators and their 
restrictions to $\cF_n$ were explained in subsection \ref{ssa.1}. 
Since the non-interacting time evolution does not mix tensor factors, one has 
\begin{align*}
\int_0^t \! ds \, O_n(s)_0 & =  \mybinom[0.85]{n}{m}
\int_0^t \! ds \, \big( m \, O_m \otimes_s 
\underbrace{1 \otimes_s \cdots \otimes_s 1}_{n-m} \big)(s)_0 \\ 
& = \mybinom[0.85]{n}{m}  \Big( \int_0^t \! ds \, m \, O_m(s)_0 \Big) \otimes_s 
\underbrace{1 \otimes_s \cdots \otimes_s 1}_{n-m}  
\in \fK_n \, .
\end{align*}
Applying $\kappa_n$, cf.\ relation \eqref{e3.3}, one obtains
\begin{align*}
\kappa_n \Big(\! \! \int_0^t \! ds \, O_n(s)_0 \Big)
& = \big((n-m)/n \big) \, \mybinom[0.85]{n}{m}
 \Big( \! \! \int_0^t \! ds \, m \, O_m(s)_0 \Big) \otimes_s 
\underbrace{1 \otimes_s \cdots \otimes_s 1}_{n-m-1}  \\
& = \mybinom[0.85]{n-1}{m}
 \Big( \! \! \int_0^t \! ds \, m \, O_m(s)_0  \Big) \otimes_s 
\underbrace{1 \otimes_s \cdots \otimes_s 1}_{n-m-1} \\
& = \int_0^t \! ds \, O_{n-1}(s)_0 \, ,
\end{align*}
completing the proof.
\qed \end{proof}

In the next step we extend this result to operators 
$O_n$, which are sandwiched 
between elements of $\fK_n$ and are acted upon by the maps $\beta_{g,n}$, 
defined in Eqn.~\eqref{ea.3}. 

\begin{lemma} \label{la.12}
Let $n \in \NN_0$ and let 
$O_n \doteq O \upharpoonright
\cF_n$, $n \in \NN_0$, be the restriction of a second quantized one-
or two-particle operator with properties given in the preceding
lemma. Then

\medskip
(i) for $K_n', K_n'' \in \fK_n$ one has 
$$
\kappa_n\Big( \int_0^t \! ds \, \big(K_n' O_n K_n''\big)(s)_0 \Big)
= \int_0^t \! ds \, \big( \kappa_n(K_n') \,  O_{n-1} \, \kappa_n(K_n'') 
\big)(s)_0 \, , 
$$

\medskip
(ii) for $K_{n-1}', K_{n-1}'' \in \fK_n$ 
and any normalized \mbox{$g \in L^2(\RR^s)$} one has 
\begin{align*}
\kappa_n \Big( & \int_0^t \! ds \,
\beta_{g,n}(K_{n-1}' O_{n-1} K_{n-1}''\big)(s)_0 \Big) \\
& = \int_0^t \! ds \, \beta_{g,n-1}\big( 
\kappa_{n-1}(K_{n-1}') \,  O_{n-2} \, \kappa_n(K_{n-1}'') 
\big)(s)_0 \, , \quad t \in \RR \, .
\end{align*}
The integrals in (i) and (ii) are elements of $\fK_n$, 
respectively~$\fK_{n-1}$, cf.\ Lemmas~\eqref{la.11} and~\eqref{la.7}.
\end{lemma}
\begin{proof}
(i) Let $s \mapsto \lambda_n(s)$, $s \in \RR$, be the function,
having values in 
normal linear maps on $\cB(\cF_n)$, which is given by
$$
\lambda_n(s)(B_n) \doteq (K_n')(s)_0 \, B_n \, (K_n'')(s)_0 \, , 
\quad B_n \in \cB(\cF_n) \, .
$$ 
In view of the norm continuous action of the time translations on $\fK_n$,
this function is norm continuous and 
maps $\fK_n$ into itself. Next, the function $s \mapsto O_n(s)_0$
is s.o.~continuous and $\int_0^t \! ds \, O_n(s)_0 \in \fK_n$ according to the
preceding lemma, $t \in \RR$. Thus, by Lemma \eqref{la.6}, the 
function $t \mapsto \int_0^t \! ds \, \big( K_n' \, O_n \, K_n'' \big)(s)_0$
has values in $\fK_n$. Moreover, it can be approximated  
in the limit of large $m \in \NN$ by finite sums of the form 
\begin{align*}
& \sum_{l = 1}^m \lambda(lt/m)\Big(
\int_{(l-1)t/m}^{lt/m} \! ds \, O_n(s)_0 \Big) \\ 
& = 
\sum_{l = 1}^m K_n'(lt/m)_0  \Big( \int_{(l-1)t/m}^{lt/m} \! ds \, 
O_n(s)_0 \Big)  K_n''(lt/m)_0 \, . 
\end{align*}
We apply to this equality the homomorphism $\kappa_n$, taking
into account that $\kappa_n \big(K_n(s)_0 \big) = 
\big(\kappa_n (K_n)\big)(s)_0$, $s \in \RR$, 
since the non-interacting dynamics does not mix tensor factors 
of the operators $K_n \in \fK_n$. 
So by Lemma \ref{la.11}, we obtain
\begin{align*}
\kappa_n & \Big( \sum_{l = 1}^m \lambda(lt/m)\Big(\int_{(l-1)t/m}^{lt/m} \! \! ds \, 
O_n(s)_0 \Big) \Big) \\
& = \sum_{l = 1}^m \kappa_n(K_n')(lt/m)_0 \, \Big( \int_{(l-1)t/m}^{lt/m} \! \! ds \, 
O_{n-1}(s)_0 \Big) \, \kappa_n(K_n'')(lt/m)_0 \, .
\end{align*}
Since $\kappa_n$ is norm continuous, we can 
proceed in the latter equality to the limit $m \rightarrow \infty$,
giving the first statement of the lemma.

\medskip
(ii) For the proof of the second statement, we make use of the fact
that for any $K_{n-1} \in \fK_{n-1}$ there exists some operator 
$A \in \obfA$ such that $A \upharpoonright \cF_{n-1} = K_{n-1}$,
cf.~\cite[Lem.~3.3]{Bu1}. Conversely, 
given an observable $A$ and any $l \in \NN_0$, there 
exists some operator $K_l \in \fK_l$ such that $K_l = A \upharpoonright \cF_l$, 
cf.~\cite[Lem.~3.2]{Bu1}. Moreover, the operators 
$K_l$ satisfy the coherence condition
$\kappa_l(K_l) = K_{l-1}$, cf.~\cite[Lem.~3.4]{Bu1}. 
Since $\beta_{g}(A) \in \obfA$, the latter fact implies     
$L_l \doteq \beta_{g,l}(K_{l-1}) = \beta_{g}(A) \upharpoonright \cF_l
\in \fK_l$. Thus  
\begin{align*}
\kappa_l\big(\beta_{g,l}(K_{l-1})\big)  =
\kappa_l\big(L_l) = L_{l-1} = \beta_{g,l-1}(K_{l-2}\big) 
= \beta_{g,l-1}\big(\kappa_l(K_{l-1})\big) \, ,
\end{align*} 
leading to the intertwining relations 
$\kappa_l \, \scirc \, \beta_{g,l} = \beta_{g,l-1} \, \scirc \, \kappa_l$,
$l \in \NN_0$. 

\medskip 
Bearing in mind that we are dealing with the non-interacting time evolution, 
we can proceed now as in the proof of Lemma \ref{la.7}, giving 
$$
\big( \beta_{g,n}(K'_{n-1} \, O_{n-1} \, K''_{n-1}) \big)(s)_0 =
\beta_{g(s),n} \big((K'_{n-1} \, O_{n-1} \, K''_{n-1})(s)_0 \big) \, ,
$$
where $g(s) = e^{i s \sbP^{\, 2}_{\kappa}} \, g$, $\, s \in \RR$. 
The function $s \mapsto \beta_{g(s),n}$ of normal linear maps on 
$\cB(\cF_{n-1})$ is norm continuous and maps 
$\fK_{n-1}$ into $\fK_n$, cf.\ Lemma~\ref{la.4}.
Moreover, $\int_0^t \! ds \, \big( K_{n-1}' O_{n-1} K_{n-1}'' \big)(s)_0 
\in \fK_{n-1}$
according to the preceding step. Hence, by Lemma~\ref{la.6},
the integrals 
$\int_0^t  \! ds \, \big(\beta_{g,n} (K_{n-1}' O_{n-1} K_{n-1}'') \big)(s)_0$
are elements of $\fK_n$, $t \in \RR$. They can be approximated in norm 
in the limit of large $m \in \NN$ by the sums 
$$
\sum_{l = 1}^m \beta_{g \, (lt/m), \, n} \Big( 
\int_{(l-1)t/m}^{lt/m} \! ds \, (K'_{n-1} \, O_{n-1} \, K''_{n-1})(s)_0  
\Big) \, .
$$
Applying to this relation the homomorphism $\kappa_n$, 
the initial remarks imply 
\begin{align*}
\kappa_n &  \Big(\sum_{l = 1}^m \beta_{g \, (lt/m), \, n} \Big( 
\int_{(l-1)t/m}^{lt/m} \! ds \, (K'_{n-1} \, O_{n-1} \, K''_{n-1})(s)_0  
\Big) \Big)  \\
 = \sum_{l = 1}^m &  \, \beta_{g \, (lt/m), \, n-1} \Big( 
\int_{(l-1)t/m}^{lt/m} \! ds \, (\kappa_n(K'_{n-1}) \, O_{n-2} \, 
\kappa_n(K''_{n-1}) \Big)(s)_0 \, . 
\end{align*}
Proceeding again to the limit of large $m$, this establishes
the second part of the statement.
\qed \end{proof}

The statements of the preceding two lemmas 
remain true if one replaces the non-interacting dynamics by the 
interacting one. The proof of this 
assertion is accomplished by the following result.

\begin{lemma} \label{la.13}
Let $B_n \in \cB(\cF_n)$, $B_{n-1} \in \cB(\cF_{n-1})$,   
such that \mbox{$\int_0^t \! ds \, B_n(s)_0 \in \fK_n$} and \ 
$\kappa_n \big(\int_0^t \! ds \, B_n(s)_0 \big) 
= \int_0^t \! ds \, B_{n-1}(s)_0$, $n \in \NN_0$. \
Then one has for the interac\-ting dynamics \
$\int_0^t \! ds \, B_n(s) \in \fK_n$  and \
$\kappa_n \big(\int_0^t \! ds \, B_n(s) \big) 
= \int_0^t \! ds \, B_{n-1}(s)$, $t \in \RR$.  
\end{lemma}
\begin{proof}
As was shown in Lemma~\ref{la.9}, the condition  
$\int_0^t \! ds \, B_n(s)_0 \in \fK_n$ implies that 
$$
\int_0^t \! ds \, B_n(s) 
= \int_0^t \! ds \, \lambda_n(s) 
\big( B_n(s)_0 \big)  \in \fK_n \, , \quad t \in \RR \, , 
$$ 
where $\lambda_n(s) = \ad e^{isH_n} e^{-isH_{0,n}} = 
\alpha_n(s) \, \scirc \, \alpha_{0,n}(-s)$, $s \in \RR$.  
In view of the properties of the function
$s \mapsto \lambda_n(s)$, established in the proof of
Lemma~\ref{la.9}, and the anticipated properties of 
$s \mapsto B_n(s)_0$, we can apply again
Lemma~\ref{la.6}. It implies that 
the integral on the right hand side of the above 
equality can be approximated in the limit of large $m \in \NN$ 
in norm by the sums 
$$
\sum_{l = 1}^m \lambda_n(lt/m)\Big(\! \! 
\int_{(l-1)t/m}^{lt/m} \! ds \, B_n(s)_0 \Big) \, . 
$$
We apply to these sums the homomorphism $\kappa_n$, making use of 
Lemma~4.5 and the appendix in \cite{Bu1} according to which 
$$
\kappa_n \, \scirc \, \lambda_n(s)
= \kappa_n \, \scirc \, \alpha_n(s) \, \scirc \, \alpha_{0,n}(-s)
= \alpha_{n-1}(s) \, \scirc \, \alpha_{0,n-1}(-s) \, \scirc \, \kappa_n
=  \lambda_{n-1}(s) \, \scirc \, \kappa_n \, .
$$
Thus we obtain 
\begin{align*}
\kappa_n & \Big(\! \sum_{l = 1}^m \lambda_n(lt/m)\Big(\! 
\int_{(l-1)t/m}^{lt/m} \! ds \, B_n(s)_0 \Big) \Big) \\ 
& = \sum_{l= 1}^m  \lambda_{n-1}(lt/m) \, \scirc \, \kappa_n
\Big(\! 
\int_{(l-1)t/m}^{lt/m} \! ds \, B_n(s)_0 \Big) \\ 
& = \sum_{l = 1}^m  \lambda_{n-1}(lt/m) \Big(\! \int_{(l-1)t/m}^{lt/m} \! ds \, 
B_{n-1}(s)_0 \Big) \, .
\end{align*}
Bearing in mind that $\kappa_n$ is norm continuous, the statement
then follows in the limit of large $m$ by the norm convergence 
of the sums.
\qed \end{proof}

With the help of the precding three lemmas we can determine now the action
of $\kappa_n$ on 
integrals involving the function $s \mapsto C_{f,n}(s)$
in the Dyson expansion \eqref{ea.4}, $n \in \NN$. Recall that 
$C_{f,n} = A_{f,n} + B_{f,n}$,
cf.\ Proposition~\ref{pa.5}, where
\begin{equation} \label{ea.6}
A_{f,n} = \hat{O}_{f,n} \, N_{f,n}^{-1} E_{f,n} +
\beta_{f,n}\big(\check{O}_{f,n-1} -
  \sigma_{f,n-1}(\check{O}_{f,n-1})\big) \, .
\end{equation}
Here $\hat{O}_{f,n}, \check{O}_{f,n-1}$ are the restrictions to
$\cF_n$, respectively $\cF_{n-1}$, of sums of second quantizations
of compact one- and two-particle operators. Furthermore,
\begin{equation}  \label{ea.7}
  B_{f,n} = \hat{V}_{f,n} \, N_{f,n}^{-1} E_{f,n} +
\beta_{f,n}\big(\check{V}_{f,n-1} -
  \sigma_{f,n-1}(\check{V}_{f,n-1})\big) \, ,  
\end{equation}
where $\hat{V}_{f,n}, \check{V}_{f,n-1}$ are the restrictions to
$\cF_n$, respectively $\cF_{n-1}$, of the second quantizations
of the localized pair potential $V$, cf.\ Lemmas
\ref{la.2} and \ref{la.3}. We then have the following
result involving the interacting dynamics. 

\begin{lemma} \label{la.14}
Let  $n \in \NN$, let $C_{f,n} = A_{f,n} + B_{f,n}$ 
be the operator given above, 
and let $s \mapsto D_n(s)$ be a norm continuous
function with values in $\fK_n$. Then 
$$
\kappa_n \Big(\int_0^t \! ds \, C_{f,n}(s) \, D_n(s) \Big)
= \int_0^t \! ds \, C_{f,n-1}(s) \, \kappa_n\big(D_n(s)\big) \, , 
\quad t \in \RR \, ,
$$
where the integrals have values in $\fK_n$, respectively $\fK_{n-1}$. 
\end{lemma}
\begin{proof}
For the proof that the integrals in this lemma have values in 
$\fK_n$, respectively $\fK_{n-1}$, we make use of Lemma 
\ref{la.6}: the function 
$s \mapsto \lambda_n(s)$ of normal linear maps 
on $\cB(\cF_n)$, given by 
$\lambda_n(s)(B_n) \doteq B_n \, D_n(s)$, $B_n \in \cB(\cF_n)$, is continuous 
in norm and maps $\fK_n$ into itself; and, 
as was shown in the proof of Proposition \ref{pa.10}, 
$\int_0^t \! ds \, C_{f,n}(s) \in \fK_n$.
Thus the first integral in the statement is an element of
$\fK_n$. The same argument applies to the second integral 
since $s \mapsto \kappa_n\big(D_n(s)\big) \in \fK_{n-1}$ is norm continuous
and $\int_0^t \! ds \, C_{f,n-1}(s) \in \fK_{n-1}$. 

\medskip
In order to determine the action of $\kappa_n$,
we first restrict attention to the constant function
$s \mapsto D_n(s) \doteq 1 \upharpoonright \cF_n$, \ie  
the integral $\int_0^t \! ds \, C_{f,n}(s)$. 
In the contributions \eqref{ea.6} and \eqref{ea.7}
to this integral, there appear the operators 
$N_{f,n}^{-1} E_{f,n}$ and $(\one_{n-1} + N_{f,n-1})^{\pm 1/2}$.
Putting $l=n, n-1$,
these are bounded functions $b(N_{f,l}) \in \fK_l$ of the operators 
$N_{f,l} \in \fK_l$,  
which in turn are restrictions to $\cF_l$ of 
the second quantization $N_f$ of the 
one-dimensional projection $E_{f,1}$ on $\cF_1$. 
Since the maps $\kappa_l$ are homomorphisms,
it follows that $\kappa_l\big(b(N_{f,l})\big) = b(N_{f,l-1})$.
Furthermore, as was shown in Lemma \ref{la.11}, one has
$\kappa_n \, \scirc \, \beta_{f,n} =
\beta_{f,n-1} \, \scirc \, \kappa_n$. Finally,
the operators $\hat{O}_{f,n}, \check{O}_{f,n-1}$ and
$\hat{V}_{f,n}, \check{V}_{f,n-1}$ in 
relations  \eqref{ea.6} and \eqref{ea.7}
are of the type of operators $O$ considered in Lemma \ref{la.11},
cf.\ also Lemma \ref{la.8}. 
Thus, Lemmas \ref{la.12} and \ref{la.13} apply 
to the function $s \mapsto C_{f,n}(s) = A_{f,n}(s) + B_{f,n}(s)$. 
Whence, making also use of the preceding relations,  we arrive at 
$$
\kappa_n \Big(\int_0^t \! ds \, C_{f,n}(s)\Big) =
\int_0^t \! ds \, C_{f,n-1}(s) \, , \quad t \in \RR \, .
$$

Let us turn now to the case of arbitrary 
norm continuous functions $s \mapsto D_n(s)$ with values in $\fK_n$.
Adopting the notation in the beginning of this proof, we have
$$
\int_0^t \! ds \, C_{f,n}(s) \, D_n(s) = 
\int_0^t \! ds \, \lambda_n(s)\big(C_{f,n}(s)\big) \, .
$$
According to Lemma \ref{la.6}, 
the latter integral can be approximated in norm in the
limit of large $m \in \NN$ by
$$
\sum_{l = 1}^m 
\lambda_n(lt/m)\Big(\int_{(l-1)t/m}^{lt/m} \! ds \, C_n(s) \Big) =
\sum_{l = 1}^m 
\Big( \int_{(l-1)t/m}^{lt/m} \! ds \, C_n(s) \Big) D_n(lt/m) \, .
$$
Applying to the expression on the right hand 
side of this equality the homomorphism $\kappa_n$, we obtain
$$
\sum_{l = 1}^m 
\Big( \int_{(l-1)t/m}^{lt/m} \! ds \, C_{n-1}(s) \Big) \, 
\kappa_n\big(D_n(lt/m)\big) \, , 
$$
where we made use of the result obtained in the preceding step.
Since the function $s \mapsto \kappa_n\big(D_n(s)\big) \in \fK_{n-1}$
is norm continuous and $\int_0^t \! ds \, C_{n-1}(s) \in \fK_{n-1}$,
we can proceed in the latter sum again to the limit of large 
$m$. By Lemma \ref{la.6}, we thereby   
arrive at the integral on the right hand side 
of the equality in the statement
of the lemma, completing its proof. \qed \end{proof}

With the help of the preceding lemma we can establish now
the coherence condition for the operators $\Gamma_{f,n}(t)$,
which, according to Proposition \ref{pa.10}, are elements of
$\fK_n$, $n \in \NN$. 

\begin{proposition} \label{pa.15} 
Let $n \in \NN$ and let 
$\Gamma_{f,n}(t) \in \fK_n$ 
be the operators, given in Eqn.~\eqref{ea.4}.
Then $\kappa_n\big(\Gamma_{f,n}(t)\big) = \Gamma_{f,n-1}(t)$. 
\end{proposition}
\begin{proof}
We make use again of the Dyson 
expansion \eqref{ea.4} and show that the multiple integrals
$D_{n,k}(s)$, $k \in \NN$, 
involving the operator $C_{f,n}$, cf.\ Eqn.~\eqref{ea.5}, are mapped by
$\kappa_n$ into corresponding integrals, where $C_{f,n}$ is
replaced by $C_{f,n-1}$
and $D_{n,k}(s)$ by
$D_{n-1,k}(s)$, $s \in \RR$. The statement 
then follows from the norm convergence of
the series. For its proof we make use of the inductive
argument given in the proof of 
Proposition~\ref{pa.10}. We have shown in the preceding lemma 
that 
$$
\kappa_n\big(D_{n,1}(t)\big) = 
\kappa_n \Big( \int_0^t \! ds \, C_n(s) \Big) 
= \int_0^t \! ds \, C_{n-1}(s) = D_{n-1,1}(t) \, , 
\quad n \in \NN  \, .
$$ 
Assuming that the analogous relation holds for 
the $k$-fold integrals of $C_n$, we represent 
the $(k+1)$-fold integral in the form 
$t \mapsto D_{n,k+1}(t) = \int_0^t \! ds \, C_n(s) \, D_{n,k}(s)$,
where $s \mapsto D_{n,k}(s) \in \fK_n$ is norm continuous. Thus it follows
from the preceding lemma that 
\begin{align*}
\kappa_{n}\big(D_{n,k+1}(t)\big) & =
\kappa_{n} \Big( \int_0^t \! ds \, C_n(s) \, D_{n,k}(s) \Big) \\
& = \int_0^t \! ds \, C_{n-1}(s) \, \kappa_n\big(D_{n,k}(s)\big)
= \int_0^t \! ds \, C_{n-1}(s) \, D_{n-1,k}(s) \, ,
\end{align*}
where in the last equality we made use of the induction hypothesis.
This establishes the coherence condition.
\qed \end{proof}

Let us summarize the results of this appendix. In order to
prove Theorem~\ref{t5.3}, we have analyzed the
properties of the operators (intertwiners between morphisms) 
$\Gamma_f(t) E_f$, $t \in \RR$, which 
were defined in equation \eqref{e5.2}. Since these operators
commute with the particle number operator $N$, we could proceed to
their restrictions 
$\Gamma_{f,n}(t) = \Gamma_f(t) E_f \upharpoonright \cF_n$, $n \in \NN_0$.  
We have shown in  Proposition \ref{pa.10} that 
$\Gamma_{f,n}(t) \in \fK_n$, and from Proposition \ref{pa.15}
we know that $\kappa_n\big(\Gamma_{f,n}(t)\big) = \Gamma_{f,n-1}(t) \in \fK_{n-1}$. 
Since $\Gamma_f(t) E_f$ is a bounded operator on $\cF$, this implies 
$\Gamma_f(t) E_f \in \obfA$. Hence 
$$
\ad e^{itH} (W_f) = \Gamma_f(t) \, W_f \in \obfF \, , \quad t \in \RR \, , 
$$
cf.\ also Lemma \ref{l4.1}. 
Since the observable algebra $\obfA$ is stable under the time 
translations \cite[Thm.~4.6]{Bu1} and the field algebra $\obfF$ is generated 
by $\obfA$ and the tensors $W_f$, $W_f^*$, this proves that 
$\ad e^{itH} (\obfF) = \obfF$, $t \in \RR$.

\medskip
It also follows  
from Proposition \ref{pa.10} that the functions 
$t \mapsto \Gamma_{f,n}(t) \in \fK_n$ are norm continuous, $n \in \NN_0$.
Hence 
$t \mapsto \Gamma_f(t) E_f$ is lct-continuous. 
It implies that the time translated tensors
$ t \mapsto \ad e^{itH} (W_f) = \Gamma_f(t) \, W_f$
are lct-continuous, $t \in \RR$.
Since the time translations act lct-continuously on the 
observable algebra $\obfA$, cf.\ \cite[Thm.\ 4.6]{Bu1}, 
and the finite polynomials in 
the basic tensors, multiplied with observables, 
are norm dense in the field algebra $\obfF$,
this establishes the lct-continuity of the time tranlsations 
$t \mapsto \ad e^{itH}$ on $\obfF$.

\medskip
For the proof of existence of an lct-dense 
sub-C*-algebra $\obfF_0 \subset \obfF$
on which the time translations act norm-continuously, we proceed as in 
Theorem \ref{t5.2}. Let 
$F_m \in \obfF$ be any tensor, $m \in \ZZ$, 
let $t \mapsto k(t)$ be any continuous function 
on $\RR$ with compact support, and consider the 
integral $F_m(k) \doteq \int \! dt \, k(t) \, \ad e^{itH}(F_m)$.
The resulting function $t \mapsto \ad e^{itH} \big(F_m(k)\big)$
is, due to the regularization, norm-continuous on the full Fock space $\cF$. 
It has values in~$\obfF$ and the C*-algebra $\obfF_0$ generated by 
the operators $F_m(k)$ is lct-dense in~$\obfF$. For the proof of the  
latter assertions, we make use of the fact that for any 
$l \in \NN_0$   
the restrictions of the gauge invariant operators, $m \leq 0$,    
\begin{align*}
G_l(t)  
\doteq W_f^m \, \ad e^{itH}(F_m)  \upharpoonright \cF_l  
= \ad e^{itH}\big( \big(\ad e^{-itH}(W_f) \big)^m \, F_m \big) \,
\upharpoonright \cF_l \in \fK_l 
\end{align*}
are norm continuous in $t \in \RR$ and satisfy 
$\kappa_l\big(G_l(t)\big) = G_{l-1}(t)$. 
So their integrals $G_l(k) = \int \! dt \, k(t) \, G_l(t)$ 
exist in the norm topology,
hence are elements of~$\fK_l$, and 
$\kappa_l\big(G_l(k)\big) = G_{l-1}(k)$. The coherent sequence 
$\{ G_l(k) \}_{l \in \NN_0}$ 
defines some bounded element $G(k) \in \obfA$ and
$F_m(k) = W_f^{* \, m} \, G(k) \in \obfF$, as claimed. A similar argument
applies to tensors with $m \geq 0$. 

\medskip
The assertion that the C*-algebra $\obfF_0$ is lct-dense in $\obfF$ follows
from the fact that one can proceed in the integrals 
$G_l(k) = \int \! dt \, k(t) \, G_l(t)$ with $k$ to the Dirac
measure, whereby the sequence $G_l(k) \in \fK_l$ converges in norm to $G_l$,    
$l \in \NN_0$. This establishes the lct-density of~$\obfF_0$
in $\obfF$ and completes the proof of Theorem \ref{t5.3}.

\newpage
\noindent {\Large \bf Acknowledgement} \\[1mm]
I would 
like to thank Dorothea Bahns and the Mathematics Institute of the   
University of G\"ottingen for their generous hospitality.   
I am also grateful to Mathieu Lewin for explaining to me his approach 
to the algebraic treatment of Bosonic systems and to Wojciech Dybalski 
for discussions on an extension of the present 
results to non-gauge invariant dynamics.

\end{document}